 \newif\ifabstract
 \abstracttrue
 \newif\iffull
 \ifabstract \fullfalse \else \fulltrue \fi

\iffull 
\documentclass[nolineno]{llncs}
\usepackage{fullpage}
\fi
\ifabstract
\documentclass{llncs} 
\fi
\usepackage{amssymb,amsmath,complexity}
\usepackage[misc]{ifsym}
\usepackage{hyperref}
\usepackage{nicefrac}
\usepackage{url}
\usepackage{color}
\usepackage{pdflscape}
\usepackage[american]{babel}
\usepackage{graphicx}
\usepackage{version}
\usepackage{subfigure} 
\usepackage[disable]{todonotes}
\usepackage{pifont}

\usepackage{amsthm}

\renewcommand{\orcidID}[1]{\href{https://orcid.org/#1}{\includegraphics[scale=.03]{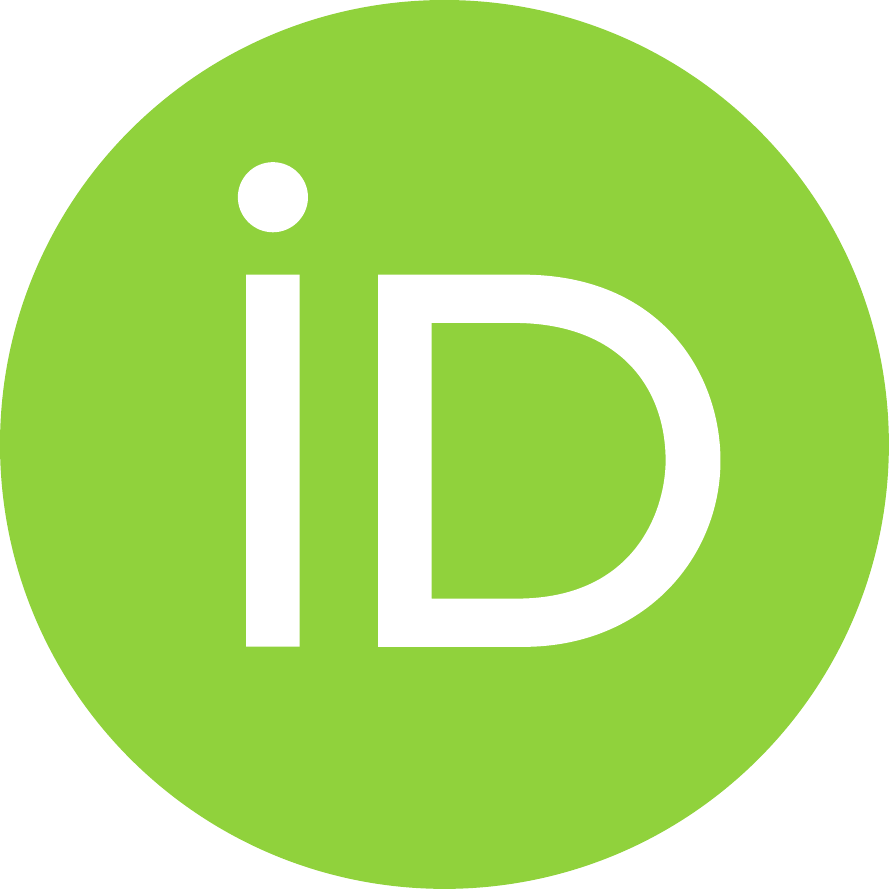}}}

\newcounter{section-preserve}
\newcounter{theorem-preserve}
\newcommand{\blank}[1]{}
\newtoks\magicAppendix
\magicAppendix={}
\newtoks\magictoks
\newif\iflater
\laterfalse
\long\def\later#1{\magictoks={#1}%
	\edef\magictodo{\noexpand\magicAppendix={\the\magicAppendix 
			\the\magictoks%
	}}
	\magictodo}
\long\def\both#1{\magictoks={#1}%
	\edef\magictodo{\noexpand\magicAppendix={\the\magicAppendix 
			\noexpand\setcounter{theorem-preserve}{\noexpand\arabic{theorem}}%
			\noexpand\setcounter{theorem}{\arabic{theorem}}%
			\noexpand\setcounter{section-preserve}{\noexpand\arabic{section}}%
			\noexpand\setcounter{section}{\arabic{section}}%
			\noexpand\let\noexpand\oldsection=\noexpand\thesection
			\noexpand\def\noexpand\thesection{\thesection}
			\noexpand\let\noexpand\oldlabel=\noexpand\label
			\noexpand\let\noexpand\label=\noexpand\blank
			\the\magictoks%
			\noexpand\setcounter{theorem}{\noexpand\arabic{theorem-preserve}}%
			\noexpand\setcounter{section}{\noexpand\arabic{section-preserve}}%
			\noexpand\let\noexpand\thesection=\noexpand\oldsection
			\noexpand\let\noexpand\label=\noexpand\oldlabel
	}}
	\magictodo
	\the\magictoks}
\def\magicappendix{\latertrue \the\magicAppendix}

\def\abstractlater#1{\ifabstract\later{#1}\fi}

\long\def\onlyfull#1{}
\long\def\onlyabstract#1{#1}  

\iffull
\long\def\both#1{{\color{blue}#1}}
\long\def\later#1{{\color{brown}#1}}
\long\def\onlyfull#1{{\color{green!60!black}#1}}
\long\def\onlyabstract#1{}
\def\magicappendix{}
\fi

\graphicspath{{figures/}}

\newtheorem{observation}{Observation}

\newcommand{\skel}[1]{{\ensuremath{\overline{#1}}}}
\newcommand{\dual}[1]{{\ensuremath{\overline{#1^*}}}}

\newcommand{\leaveout}[1]{}

\begin{document}

\setcounter{page}{1}

\title{Optimal-area visibility representations of outer-1-plane graphs\thanks{Work of TB supported by NSERC; FRN RGPIN-2020-03958. Work of GL supported by MIUR, grant 20174LF3T8 ``AHeAD: efficient Algorithms for HArnessing networked Data''. Work of FM supported by Dipartimento di Ingegneria,  Universit\`a degli Studi di Perugia, grant RICBA19FM.}}

\author{Therese Biedl\inst{1}\orcidID{0000-0002-9003-3783}({\small \Letter}) \and
Giuseppe Liotta\inst{2}\orcidID{0000-0002-2886-9694} \and \\
Jayson Lynch\inst{1}\orcidID{0000-0003-0801-1671}
\and 
Fabrizio Montecchiani\inst{2}\orcidID{0000-0002-0543-8912}}
\authorrunning{T. Biedl et al.}
%
\institute{David R.~Cheriton School of Computer Science, University of Waterloo, \email{\{biedl,jayson.lynch\}@uwaterloo.ca} \and
Department of Engineering, University of Perugia, \email{\{giuseppe.liotta,fabrizio.montecchiani\}@unipg.it}}
\maketitle  

\begin{abstract}

This paper studies optimal-area visibility representations of $n$-vertex outer-1-plane graphs, i.e. graphs with a given embedding where all vertices are on the boundary of the outer face and each edge is crossed at most once. We show that any graph of this family admits an embedding-preserving visibility representation whose area is $O(n^{1.5})$ and prove that this area bound is worst-case optimal. We also show that $O(n^{1.48})$ area can be achieved if we represent the vertices as L-shaped orthogonal polygons or if we do not respect the embedding but still have at most one crossing per edge. We also extend the study to other representation models and, among other results, construct asymptotically optimal $O(n\, pw(G))$ area  bar-1-visibility representations, where $pw(G)\in O(\log n)$ is the pathwidth of the outer-1-planar graph $G$.

\end{abstract}

\section{Introduction} \label{se:introduction}

Visibility representations are one of the oldest topics studied in
graph drawing:  Otten and van Wijk showed in 1978 that every planar graph 
has a visibility representation~\cite{OW78}. 
A {\em rectangle visibility representation} consists of an assignment of disjoint axis-parallel boxes to vertices, and axis-parallel segments to edges in such a way that edge-segments end at the vertex boxes of their endpoints and do not intersect any other vertex boxes. (They can hence be viewed as lines-of-sight, though not every line-of-sight needs to give rise to an edge.)  

Vertex-boxes are permitted to be degenerated into a segment or a point (in our pictures we thicken them slightly for readability).
The construction by Otten and van Wijk
is also {\em uni-directional} (all
edges are vertical) and
all vertices are {\em bars} (horizontal segments
or points).  Multiple other papers studied uni-directional bar-visibility
representations 
and showed that these exist if and only if the
graph is planar~\cite{DHVM83,Wis85,RT86,TT86}. 

Unless otherwise specified, we assume throughout this paper that any visibility
representation $\Gamma$  (as well as the generalizations we list below) are on an \emph{integer grid}. 
This means that all corners of vertex polygons, as well as all \emph{attachment points} (places where edge-segments end at vertex polygons) have integer coordinates. 
The \emph{height} [\emph{width}] of  $\Gamma$ is the number of grid rows [columns] that intersect $\Gamma$. The \emph{area} of $\Gamma$ is its width times its height. 
Any visibility representation can be assumed to have 
\onlyfull{width and height $O(n)$, hence }area $O(n^2)$ (see also Obs.~\ref{obs:widthHeightUpperTrivial}).
Efforts have been made to obtain small constants factors%
\iffull
, with the current best bound (to our knowledge) area $\tfrac{4}{3}n^2
+o(n^2)$~\cite{FLLY07} and $\tfrac{1}{2}n^2+o(n^2)$ for 4-connected planar graphs~\cite{HZ08}.
\else
~\cite{FLLY07,HZ08}.
\fi

In this paper, we focus on {\em bi-directional} rectangle visibility 
representations, i.e., both horizontally and vertically drawn edges are 
allowed.    For brevity we drop `bi-directional' and `rectangle' from now
on.  
\onlyfull{Such visibility representations can exist only if the graph is the union of two planar graphs (and so has $O(n)$ edges), but this is not a characterization~\cite{HSV99}.}
Recognizing graphs that have a visibility representation is \NP-hard
\cite{Shermer96a}.  
\iffull
(If some aspects of the visibility representation
are specified, e.g.~the edge-directions or the planarization of the
drawing, then graphs that have such a representation  can
be recognized in polynomial time~\cite{BFL-DCG,SW03}.)
\fi
Planar graphs have visibility representations where the area is at most $n^2$, and $\Omega(n^2)$ area is sometimes required~\cite{FKK-GD96}.
For special graph classes, $o(n^2)$ area can be achieved, such as
$O(n{\cdot}pw(G))$ area for outer-planar graphs~\cite{Bie-WAOA12} (here
$pw(G)$ denotes the {\em pathwidth} of $G$, defined later), and $O(n^{1.5})$
area for series-parallel graphs~\cite{Bie-DCG11}.    
The latter two results do not give embedding-preserving drawings (defined below).

\paragraph{Variations of visibility representations.}  For graphs that
do not have visibility representations (or where the area-requirements are
larger than desired), other models have been introduced that are similar
but more general.    One option is to increase the dimension, 
see e.g.~\cite{BEF+93,DBLP:journals/tcs/AngeliniBKM19,ABD+18}.   We will not do this here, and instead allow more complex shapes for vertices or edges.  Define an \emph{orthogonal polygon [polyline]} to be a polygon [polygonal line] whose segments are horizontal or vertical.   We use OP as convenient shortcut for `orthogonal polygon'.  All variations that we study below are what we call \emph{OP-$\infty$-orthogonal drawings}.%
\footnote{We do not propose actually drawing graphs in this model (its readability would not be good), but it is convenient as a name for ``all drawing models that we study here''.}  
Such a drawing is an assignment of disjoint orthogonal polygons $P(\cdot)$ to vertices and orthogonal poly-lines to edges such that the poly-line of edge $(u,v)$ connects $P(u)$ and $P(v)$.  Edges can intersect each other, and they are specifically \emph{allowed} to intersect arbitrarily many vertex-polygons (hence the ``$\infty$''), but no two edge-segments are allowed to overlap each other.    The \emph{vertex complexity} is the maximum number of reflex corners in a 
\iffull
vertex-polygon (after slightly thickening polygons so that every corner of
the polygon is adjacent to points in the strict interior and ``reflex corner''
is well-defined).  The
\else
vertex-polygon, and the
\fi
\emph{bend complexity} is the maximum number of bends in an edge-poly-line.
\todo[inline]{TB: To be consistent, we should either renamed vertex complexity to polygon complexity, or bend complexity to edge complexity.  I'd prefer the former.}

One variation that has been studied is 
\emph{bar-$(k,j)$-visibility representation}, where vertices are bars, edges are vertical line segments, edges may intersect up to $k$ bars that are not their endpoints, and any vertex-bar is intersected by at most $j$ edges that do not end there.
\onlyfull{Note that bar-$(k,j)$-visibility representations are
by definition uni-directional (the concept could be generalized to 
bi-directional, but this appears not to have been studied).  }
Bar-$(k,\infty)$-visibility representations  were introduced by Dean et al.~\cite{DEG07}, and testing whether a graph has one is
\NP-hard~\cite{BHKN15}.  All 1-planar graphs have a bar-$(1,1)$-visibility
representation~\cite{Bra14,EKL+14}.  In this paper, 
we will use \emph{bar-$1$-visibility representation} 
as a convenient shortcut for ``unidirectional bar-$(1,1)$-visibility representation''.

Another variation is \emph{OP visibility representation}, 
\todo{TB: insert ``which are OP-$\infty$-orthogonal drawings''}
where edges must be horizontal or vertical segments that do not intersect vertices except at their endpoints.
OP visibility representations were introduced
by Di Giacomo et al.~\cite{DDE+18} and they exist for all 1-planar graphs.    
There are further studies, considering the vertex complexity that may be required in such drawings~\cite{DBLP:journals/comgeo/Brandenburg18,DDE+18,DBLP:journals/tcs/EvansLM16,DBLP:journals/ipl/LiottaM16,DBLP:journals/tcs/LiottaMT21}.   

Finally, there are \emph{orthogonal box-drawings}, where vertices must be boxes and edges do not intersect vertices except at their endpoints. We will not review the (vast) literature on orthogonal box-drawings (see e.g.~\cite{BK-ESA97,DBLP:conf/esa/BlasiusBR14} and the references therein), but they exist for all graphs.

All OP-$\infty$-orthogonal drawings can be assumed to have area $O(n^2)$ (assuming constant complexity and $O(n)$ edges), see also Obs.~\ref{obs:widthHeightUpperTrivial}.  
We are not aware of any prior work that tries to reduce the area to $o(n^2)$ for specific graph~classes.

\paragraph{Drawing outer-1-planar graphs. }  An \emph{outer-1-planar graph} (first defined by Eggleton~\cite{Eggleton}) is a graph that has a drawing $\Gamma$ in the plane such that all vertices are on the infinite region of $\Gamma$ and every edge has at most one crossing.  \onlyfull{(More detailed definitions are below.)  }
We will not review the (extensive) literature on their superclass of 1-planar graphs here; see e.g.~\cite{KLM17} or \cite{DBLP:journals/csur/DidimoLM19,BeyondPlanarBook} for even more related graph classes. Outer-1-planar graphs can be recognized in linear time~\cite{HEK+15,Auer16}.
 All outer-1-planar graphs are planar~\cite{Auer16}, and so can be drawn in $O(n^2)$ area, albeit not embedding-preserving.    

Very little is know about drawing outer-1-planar graphs in area $o(n^2)$.
Auer et al.~\cite{Auer16} claimed to construct planar visibility representations of area $O(n\log n)$, but this turns out to be incorrect~\cite{Bie-GD20} since some outer-1-planar graphs require $\Omega(n^2)$ area in planar drawings.     Outer-1-planar graphs do have orthogonal
box-drawings with bend complexity 2 in $O(n\log n)$ area~\cite{Bie-GD20}.

\renewcommand{\arraystretch}{1.2}
\begin{table}[b]
\iffull
\centering
\begin{tabular}{|c|c|c|c|}
\hline
drawing-style & e-p & lower bound & upper bound  \\
\hline\hline
visibility representation & \ding{51} & $\Omega(n^{1.5})$ [Thm.~\ref{thm:lowerGeneral}] & $O(n^{1.5})$ [Thm.~\ref{thm:sqrt}] \\
\hline
complexity-1 OP vis.repr.& \ding{51} & $\Omega(n \,pw(G))$ [Thm.~\ref{thm:pwLowerBounds}] & $O(n^{1.48})$ [Thm.~\ref{thm:breaking}]\\
\hline
 1-bend orth.~box-drawing & \ding{51} & $\Omega(n \,pw(G))$ [Thm.~\ref{thm:pwLowerBounds}]& $O(n^{1.48})$ [Thm.~\ref{thm:breaking}]\\
\hline
visibility representation & \ding{55} & $\Omega(n \,pw(G))$ [Thm.~\ref{thm:pwLowerBounds}]& $O(n^{1.48})$ [Thm.~\ref{thm:breaking}]\\
\hline
bar visibility representation  & \ding{51} & $\Omega(n^2)$ [Thm.~\ref{thm:lowerBVR}]& $O(n^2)$ [Thm.~\ref{thm:barVR}]\\
\hline
bar-1-visibility representation & \ding{55} & $\Omega(n \,pw(G))$ [Thm.~\ref{thm:pwLowerBounds}]& $O(n\,pw(G))$ [Thm.~\ref{thm:pathDrawing}]\\
\hline
planar visibility representation & \ding{55} & $\Omega(n (pw(G){+}\chi(G)))$ [Thm.\ref{thm:pwLowerBounds}\&\ref{thm:crossingLowerBounds}]& $O(n(pw(G){+}\chi(G)))$ [Thm.\ref{thm:pathDrawing}]\\
\hline 
complexity-4 OP vis.repr.& \ding{51} & $\Omega(n \,pw(G))$ [Thm.~\ref{thm:pwLowerBounds}] & $O(n\,pw(G)))$ [Thm.~\ref{thm:complexity4}]\\
 \multicolumn{4}{c}{\todo[inline]{if we add those results?}} \\
\hline
\end{tabular}
\else
\begin{scriptsize}
\begin{tabular}{|c|c|c|c|}
\hline
drawing-style & e-p & lower bound & upper bound  \\
\hline\hline
visibility representation & \ding{51} & $\Omega(n^{1.5})$ [Thm.~\ref{thm:lowerGeneral}] & $O(n^{1.5})$ [Thm.~\ref{thm:sqrt}] \\
\hline
complexity-1 OP vis.repr.& \ding{51} & $\Omega(n \,pw(G))$ [Thm.~\ref{thm:pwLowerBounds}] & $O(n^{1.48})$ [Thm.~\ref{thm:breaking}]\\
\hline
 1-bend orth.~box-drawing & \ding{51} & $\Omega(n \,pw(G))$ [Thm.~\ref{thm:pwLowerBounds}]& $O(n^{1.48})$ [Thm.~\ref{thm:breaking}]\\
\hline
visibility representation & \ding{55} & $\Omega(n \,pw(G))$ [Thm.~\ref{thm:pwLowerBounds}]& $O(n^{1.48})$ [Thm.~\ref{thm:breaking}]\\
\hline
bar visibility representation  & \ding{51} & $\Omega(n^2)$ [Thm.~\ref{thm:lowerBVR}]& $O(n^2)$ [Thm.~\ref{thm:barVR}]\\
\hline
bar-1-visibility representation & \ding{55} & $\Omega(n \,pw(G))$ [Thm.~\ref{thm:pwLowerBounds}]& $O(n\,pw(G))$ [Thm.~\ref{thm:pathDrawing}]\\
\hline
planar visibility representation & \ding{55} & $\Omega(n (pw(G){+}\chi(G)))$ [Thm.\ref{thm:pwLowerBounds}\&\ref{thm:crossingLowerBounds}]& $O(n(pw(G){+}\chi(G)))$ [Thm.\ref{thm:pathDrawing}]\\
\hline 
\end{tabular}
\end{scriptsize}
\fi  

\vspace{1mm}
\caption{Upper and lower bound on the area achieved in various drawing styles in this paper. The column title e-p stands for embedding-preserving, $pw(G)$ denotes the pathwidth of $G$, $\chi(G)$ denotes the number of crossings in the 1-planar embedding.}
\label{ta:overview}
\end{table}

\paragraph{Our results.}  We study visibility representations (and variants) of outer-1-planar graphs, especially drawings that preserve the given outer-1-planar embedding.  Table~\ref{ta:overview} gives an overview of all results that we achieve. As our main result, we give tight upper and lower bounds on the area of embedding-preserving visibility representations (Section~\ref{sec:lower} and \ref{sec:sqrt}): It is $\Theta(n^{1.5})$. 
We find it especially interesting that the lower bound is neither $\Theta(n\log n)$ nor $\Theta(n^2)$ (the most common area lower bounds in graph drawing results). 
\ifabstract
Also, a tight area bound is not known for embedding-preserving visibility representations of outerplanar~graphs. 
\fi
\onlyfull{
Also, the bounds match. 
In contrast, even for the very well-studied class of outer-planar graphs, no such matching bounds were known: The lower bound is $\Omega(n\log n)$ (obtained by combining a complete binary tree with a node of degree $n$, see also Observation~\ref{obs:widthHeightLowerTrivial} and Theorem~\ref{thm:pwLowerBounds})
but the smallest-area embedding-preserving visibility representations have, to our knowledge, $O(n^{1+\varepsilon})$ for arbitrary $\varepsilon>0$.  (This result has not been stated explicitly, but follows by taking an embedding-preserving straight-line drawing of width $O(n^\varepsilon)$ \cite{FPR20}, rotating it by 90$^\circ$, and then transforming it to a visibility representation with the same height and embedding \cite{Bie-GD14}.  This then has height $O(n^\varepsilon)$ and width $O(n)$.)
} 

We also show in Section~\ref{sec:breaking} that the $\Omega(n^{1.5})$ area bound can be undercut if we relax the drawing-model slightly, and show that area $O(n^{1.48})$ can be achieved in three other drawing models. 
Finally we give further area-optimal results in other drawing models in Section~\ref{sec:bar}.  To this end, we generalize a well-known lower bound using the pathwidth to \emph{all} OP-$\infty$-orthogonal drawings, and also develop an area lower bound for the planar visibility representations of outer-1-planar graphs based on the number of crossings in an outer-1-planar embedding. Then we give constructions that show that these can be matched asymptotically. We conclude in Section~\ref{sec:conclusion} with open problems.

\onlyabstract{For space reasons we only sketch the proofs of most theorems; a $(\star)$ symbol indicates that further details can be found in the appendix.}
\section{Preliminaries}
\label{sec:preliminaries}
\abstractlater{\section{Missing details from Section~\ref{sec:preliminaries}}}

We assume familiarity with standard graph drawing terminology~\cite{DBLP:books/ph/BattistaETT99}.   Throughout the paper, $n$ and $m$ denotes the number of vertices and edges.  \todo{Ideally add: ``In all drawings, no three edges cross in a point.}

A planar drawing of a graph subdivides the plane into topologically connected regions, called \emph{faces}. The unbounded region is called the \emph{outer-face}. 
\onlyfull{
A \emph{planar embedding} $\mathcal E(G)$ of a planar graph $G$ is an equivalence class of planar drawings that define the same set of circuits that bound faces. A planar embedding is uniquely defined by the circular list of the edges around each vertex together with the choice of the outer-face. The concept of 
embedding can be extended drawings with crossings as follows. }
An \emph{embedding} $\mathcal E(G)$ of a graph $G$ is an equivalence class of drawings whose \emph{planarizations} (i.e., planar drawings obtained after replacing crossing points by dummy vertices) define the same set of circuits that bound faces. An \emph{outer-1-planar drawing} is a drawing with at most one crossing per edge and all vertices on the outer-face. An \emph{outer-1-planar graph} is a graph admitting an outer-$1$-planar drawing. An \emph{outer-1-plane graph} $G$ is 
a graph with a given outer-1-planar embedding $\mathcal E(G)$. 
We use $\chi(G)$ for the number of crossings in $\mathcal E(G)$.
An outer-$1$-plane graph $G$ is \emph{plane-maximal} if it is not possible to add any uncrossed edge without losing outer-$1$-planarity or simplicity.  The \emph{planar skeleton} of an outer-$1$-plane graph $G$, denoted by $\skel{G}$, is the graph induced by its uncrossed edges.
If $G$ is plane-maximal, 
then $\skel{G}$ is a 2-connected graph whose interior faces have degree 3 or 4  \cite{DBLP:journals/ijcga/DehkordiE12}.
Let $\dual{G}$ be the weak dual of $\skel{G}$ \todo{ideally add definition of weak dual; this is not so standard} and call it the \emph{inner tree} of $G$. Since $\skel{G}$ is outer-plane, $\dual{G}$ is a tree (as the name suggests), and since each face of $\skel{G}$ has degree 3 or 4, every vertex of $\dual{G}$ has degree at most 4.  An {\em outer-$1$-path} $P$ is an outer-$1$-plane graph whose inner tree $\dual{P}$ is a path. 

Consider a graph $G$ with a fixed embedding $\mathcal{E}(G)$.
An OP-$\infty$-orthogonal drawing is
\emph{embedding-preserving} if (1) walking around each vertex-polygon we encounter the incident edges in the same cyclic order as in $\mathcal{E}(G)$, and (2) no edge crosses a vertex, and the planarization of the OP-$\infty$-orthogonal drawing has the same set of faces as $\mathcal{E}(G)$.
Note that bar-1-visibility representations by definition violate (2), but we call them embedding-preserving if (1) holds.


Our results will only consider the smaller dimension of the drawing (up to rotation the height), because the other dimension does not matter (much): 

\both{
\begin{observation}
\label{obs:widthHeightUpperTrivial}
Let $\Gamma$ be an OP-$\infty$-orthogonal drawing with constant vertex and bend complexity.  Then we may assume that the width and height is $O(n+m)$.
\end{observation}
}
\later{
\begin{proof}
Due to the restrictions on the complexity/bends, we have $O(n{+}m)$ segments in edge-polylines or vertex-polygons.
Since all segments are horizontal and vertical, we can delete empty rows/columns and so need at most one row [column] per horizontal [vertical] segment (and often less).
\end{proof}
}

\onlyfull{(Obs.~\ref{obs:widthHeightUpperTrivial} was mentioned for visibility representations in \cite{Bie-GD14}).  }

\both{
\begin{observation}
\label{obs:widthHeightLowerTrivial}
Let $\Gamma$ be an OP-$\infty$-orthogonal drawing with constant vertex complexity in a $W\times H$-grid.  Then $\max\{W,H\}\in \Omega(\text{maximum degree of $G$})$.
\end{observation}
}
\later{
\begin{proof}
After rotation we may assume that $W\geq H$.
Let $v$ be the vertex with maximum degree $\Delta$, and let $k\in O(1)$ be the complexity of polygon $P(v)$.
If $W\geq \Delta/(8k+8) \in \Omega(\Delta)$ then we are done, so assume not.  
Consider a horizontal segment $s$ of polygon $P(v)$.
This intersects at most $\Delta/(8k+8)$ columns, hence at most $\Delta/(4k+4)$ edges can end vertically at $s$.  (This bound could actually be improved to $\Delta/(2k+2)$, but this makes no difference asymptotically.)  
Polygon $P(v)$ has $2k+2$ horizontal segments, so at most $(2k+2)\cdot \Delta/(4k+4)=\Delta/2$ edges can end vertically at $P(v)$.
So there are at least $\deg(v)-\Delta/2=\Delta/2$ edges that attach horizontally  at $P(v)$.  Since $P(v)$ has $2k+2$ horizontal edges, at least one of them hence has length $\Delta/(8k+8)$ or more. Therefore $H\geq \Delta/(8k+8)>W$, a contradiction.
\end{proof}
} 

\onlyabstract{Obs.~\ref{obs:widthHeightUpperTrivial} holds because we can delete empty rows and columns (and was mentioned for visibility representations in \cite{Bie-GD14}); Obs.~\ref{obs:widthHeightLowerTrivial} holds since some  vertex-polygons must have sufficient width or height for its incident edges. $(\star)$} 

\paragraph{Remark:}
In consequence of Obs.~\ref{obs:widthHeightLowerTrivial}, if we know a lower bound $f(G)$ on the width and height of a drawing, then (after adding degree-1 vertices to achieve maximum degree $\Theta(n)$) we know that any drawing of the resulting graph $G'$ has (up to rotation) width $\Omega(n)$ and height $\Omega(f(G))$, so area $\Omega(n\, f(G))$.  This is assuming $G'$ is within the same graph class and $f(G')\in \Theta(f(G))$;
both hold when we apply this below.

\section{Lower bound on the height}
\label{sec:lower}

\abstractlater{\section{Missing details from Section~\ref{sec:lower}}}

In this section, we show that embedding-preserving
visibility representations must have height $\Omega(\sqrt{n})$ for some outer-1-plane graphs.
A crucial ingredient is a lemma that studies the case where the height of vertex-boxes is restricted.

\def\movedProof{
Fix an arbitrary embedding-preserving visibility representation $\Gamma$ of $G_{h,\ell}$ where all vertex-boxes intersect at most $h$ rows.
Consider one copy $H$ of
$H_{h,\ell}$ and the box $P(v_0)$ of $v_0$.  
We say that $H$ {\em is right of $v_0$} if some part of the right side of $P(v_0)$ belongs to the interior of the induced drawing of $H$.
Since $P(v_0)$ intersects at most $h$ rows, there are at most $h$ edges that attach horizontally at the right side of $P(v_0)$.  If $H$ is right of $v_0$, then these attachment points are entirely used up by edges from $v_0$ into $H$ and/or edges from $v_0$ to leaves that come before and after $H$ at $v_0$.    (There are $h-1$ leaves on each side of $H$, which uses up all such points since the embedding is respected.)

So at most one copy of $H_{h,\ell}$ is right of $v_0$, and symmetrically at most one
copy of $H_{h,\ell}$ is left of $v_0$.
So there exists a copy $H$ of $H_{h,\ell}$ that is neither left nor right of $v_0$, and so the interior of the induced drawing of $H$ uses only the top side of $P(v)$ or only the bottom side of $P(v)$.  Up to symmetry, we may hence assume that  all edges from $v_0$ to $H$ go downward 
from $v_0$.  Since the drawing respects the embedding, edge $(v_0,a_1)$ must be the leftmost among the edges from $v_0$ into $H$.
}  

\begin{lemma}
\label{le:lower-2}
For any $h,\ell>0$ there exists an outer-1-plane graph $G_{h,\ell}$ 
with $O(h \cdot \ell)$ vertices such that  any  embedding-preserving visibility representation in which each vertex-box intersects at most $h$ rows, has width and height $\Omega(\ell)$.
\end{lemma}
\begin{proof}
\begin{figure}[t]
\centering
\subfigure[$H_{1,\ell}$ (top), $H_{h,\ell}$ (bottom) ]{\includegraphics[width=0.44\textwidth,page=9]{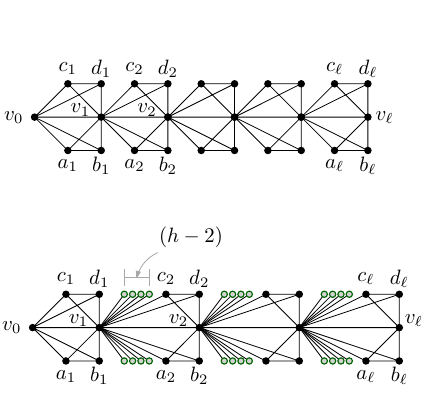}\label{fig:H1-2}}
\subfigure[$G_{h,\ell}$]{\includegraphics[width=0.44\textwidth,page=8]{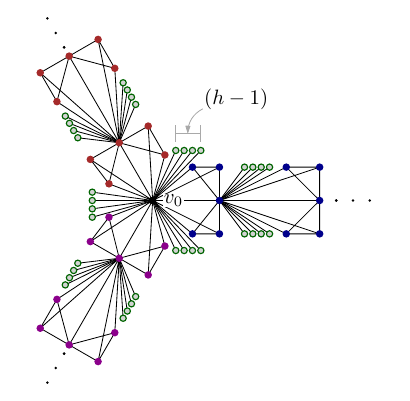}\label{fig:H3-2}}
\caption{The graph for Lemma~\ref{le:lower-2}.}
\label{fig:H-2}
\end{figure}
To build graph $G_{h,\ell}$, 
we first need a graph $H_{1,\ell}$ with $5\ell+1$ vertices, depicted in Fig.~\ref{fig:H1-2} (top).  This graph consists of
a path $v_0,\dots,v_\ell$ such that at each edge $(v_{i-1},v_i)$
(for $1\leq i\leq \ell$) there are two attached $K_4$'s
$\{v_{i-1},a_i,b_i,v_i\}$ and
$\{v_{i-1},c_i,d_i,v_i\}$
(drawn such that $(v_{i-1},b_i)$ crosses $(v_i,a_i)$ and $(v_{i-1},d_i)$ crosses $(v_i,c_i)$).
\todo{TB: reinserted this since the embedding is crucial}
  
Next define $H_{h,\ell}$ for $h\geq 2, \ell\geq 1$ by
taking $H_{1,\ell}$ and adding $2(\ell-1)(h-2)$ vertices of degree
1 (we call these {\em leaves}), as shown in Fig.~\ref{fig:H1-2} (bottom). 
Namely, at each vertex $v_i$ for $0<i<\ell$, we add $h-2$ leaves between $b_i$ and $a_{i+1}$ (in the
order around $v_i$), and another $h-2$ leaves between $c_{i+1}$ 
and $d_i$. 
Clearly graph $H_{h,\ell}$ is outer-1-planar.  

Graph
$G_{h,\ell}$ consists of three copies of $H_{h,\ell}$, with the three 
vertices $v_0$ combined into one, see also Fig.~\ref{fig:H3-2}.    
Furthermore, add $h-1$ leaves at $v_0$ between any two copies, i.e., 
between $c_1$ of one copy and $a_1$ of the next copy. 
Graph $G_{h,\ell}$ has $n=15\ell+1+6(\ell-1)(h-2)+ 3(h-1) \in \Theta(\ell h)$ 
vertices.%
\iffull
\footnote{Graph $G_{h,\ell}$ is not maximal-planar, but we can make it maximal-planar (hence 2-connected) by adding edges, and the same lower bound holds.}
\fi

\begin{figure}[ht]
\centering
\includegraphics[scale=0.8,page=10]{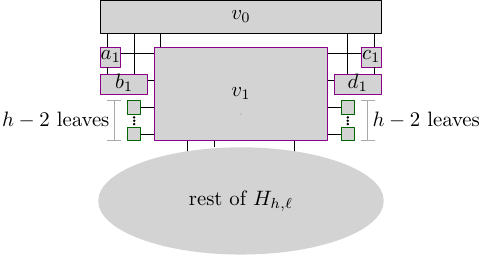}
\caption{Illustration for the proof of Lemma~\ref{le:lower-2}.}
\label{fig:H4-2}
\end{figure}

\onlyfull{
\movedProof
This is enough to prove that the  width and height is $\Omega(\ell)$,  due to the following.
}

\onlyabstract{One can argue $(\star)$ that inside any embedding-preserving visibility representation $\Gamma$ of $G_{h,\ell}$, there exists a copy of $H_{h,\ell}$ whose drawing $\Gamma_H$ satisfies (up to symmetry) the premise of the following claim.}

\begin{claim}\label{cl:gamma}
Let $\Gamma_H$ be an embedding-preserving visibility representation of $H_{h,\ell}$
such that all edges at box $P(v_0)$ go downward, with
edge $(v_0,a_1)$ leftmost among them.  Assume that 
all boxes of $\Gamma_H$  intersect at most $h$ rows. Then $\Gamma_H$ uses at least 
$\ell+1$ rows 
and $P(v_0)$ has  width at least $4\ell+1$.
\end{claim}
\begin{proof}
We proceed by induction on $\ell$.  In the base case ($\ell=1$) we have 
five vertical downward edges at $v_0$; this means that the height is at least 2 and $P(v_0)$  must have width at least 5 as required.

Now assume $\ell\geq 2$ and study the five downward edges from $v_0$ 
to $a_1,b_1,v_1,c_1,\allowbreak d_1$, see also Fig.~\ref{fig:H4-2}.
The vertical edges $(v_0,b_1)$ and $(v_0,d_1)$ are crossed by 
edges $(a_1,v_1)$ and $(v_1,c_1)$, which means that the latter two edges 
must be horizontal.    Since $(v_0,a_1)$ is leftmost, and the embedding
is preserved, edge $(a_1,v_1)$ attaches on the left side of $P(v_1)$ while $(v_1,c_1)$
attaches on the right side. 

The counter-clockwise order of edges at $v_1$ contains 
$h-1$ edges (to $b_1$ and leaves) between $a_1$ and $a_2$.
Since $P(v_1)$ intersects at most $h$ rows, and $(v_1,a_1)$ attaches
on its left side, therefore $(v_1,a_2)$ can {\em not} attach on its left side.
Likewise $(v_1,c_2)$ can {\em not} attach on the right side of $P(v_1)$.  To preserve the embedding, therefore
the edges from $v_1$
to the rest of $H_{h,\ell}$ must be drawn downward from $P(v_1)$, with $(v_1,a_2)$ leftmost.
Also observe that 
$H_{h,\ell}\setminus \{v_0,a_1,b_1,c_1,d_1\}$ 
contains a copy of $H_{h,{\ell-1}}$.  
Applying induction, there are at least $\ell$ rows below $P(v_1)$, and  $P(v_1)$ has width at least $4\ell-3$.  Adding at least one row for $P(v_0)$, and
observing that $P(v_0)$ must be at least four units wider than $P(v_1)$ proves the claim.
\end{proof}

So the claim holds, and $\Gamma_H$ (and with it $\Gamma$)
has width and height $\Omega(\ell)$.  
\end{proof}
\abstractlater{
\paragraph{Missing part of the proof of Lemma~\ref{le:lower-2}:} 
\movedProof
}

\onlyfull{
\noindent As a first consequence, we obtain a lower bound for any visibility representation.
}
\onlyabstract{
\noindent As a consequence, we obtain two lower bound for visibility representation.
}
\onlyfull{
\todo{The spacing between the theorems and the proofs below is quite bad.
This has probably something to do with the `both' and `later' that's used
here, but even throwing in \% doesn't fix it.  Probably remove those commands 
once the GD version is finalized.}
}
\both{\begin{theorem}
\label{thm:lowerGeneral}
For any $N$ there is an $n$-vertex outer-1-plane graph with $n\geq N$
such that any embedding-preserving visibility representation has area 
$\Omega(n^{1.5})$.
\end{theorem}}%
\later{%
\begin{proof}%
Set $h=\ell=\lceil \sqrt{N} \rceil$ and define graph $G$ to be 
of $G_{h,\ell}$ with $\ell^2$ further leaves added at $v_0$ (at arbitrary places); this has $n=\Theta(h\ell)=\Theta(N)$ vertices.
Fix an arbitrary embedding-preserving visibility 
representation of $G$ in a $W\times H$-grid.  Up to symmetry, assume
$W\geq H$.  By Obs.~\ref{obs:widthHeightLowerTrivial} we have $W\geq \deg(v_0)\geq \ell^2\in \Omega(n)$.
%
If any vertex-box has
height more than $h$, then this alone requires $H\geq h= \ell \in \Omega(\sqrt{n})$
and we are done.  If all vertex-boxes have height at most $h$, then by
Lemma~\ref{le:lower-2} the height is $\Omega(\ell)=\Omega(\sqrt{n})$, and again we are done.
\end{proof}

\todo[inline]{TB: It would be nice to work out the exact constant in the $\Omega$, at least for some values of $N$.    I wouldn't put the proof itself here, but some sentence such as ``One can show that with a more careful choice of $h,\ell$ the lower bound is actually $c n^{1.5}$.''  (where $c$ gets replaced by whatever).  I think $c$ is not too too small (maybe around $\tfrac{1}{25}$?), so this lower bound is not as useless as others. }
}

\onlyfull{
\noindent As a second consequence, we can get a stronger lower-bound for bar visibility representations.
}

\both{
\begin{theorem}
\label{thm:lowerBVR}
For any $N$ there is an $n$-vertex outer-1-plane graph with $n\geq N$ such that any embedding-preserving bar visibility representation has area $\Omega(n^2)$.
\end{theorem}
}
\later{
\begin{proof}
Set $h=1$ and $\ell=N$ and apply Lemma~\ref{le:lower-2} to $G_{h,\ell}$.  This graph has $\Theta(\ell)=\Theta(N)$ vertices and requires width and
height $\Omega(\ell)=\Omega(n)$ in any embedding-preserving visibility
representation where boxes intersect only one row.
\end{proof}
}

\onlyabstract{Roughly speaking, Theorem~\ref{thm:lowerGeneral} uses $G_{\sqrt{N},\sqrt{N}}$, with leaves added to have maximum degree $\Theta(n)$, while Theorem~\ref{thm:lowerBVR} uses $G_{1,N}$.  $(\star)$. The bounds then hold by Lemma~\ref{le:lower-2} and Obs.~\ref{obs:widthHeightLowerTrivial}.}

\abstractlater{Figure~\ref{fig:linear} shows drawings of the lower-bound graph with linear area in other drawing-models.}
\later{
\begin{figure}[p]
\centering
\subfigure{\includegraphics[width=0.65\textwidth,page=11]{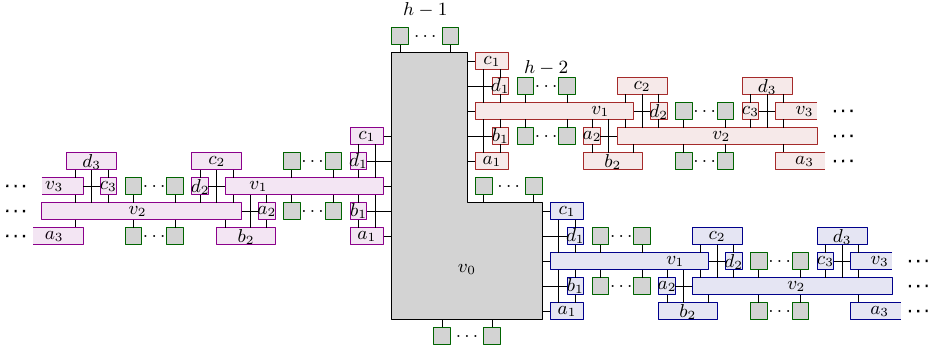}}
\subfigure{\includegraphics[width=0.65\textwidth,page=5]{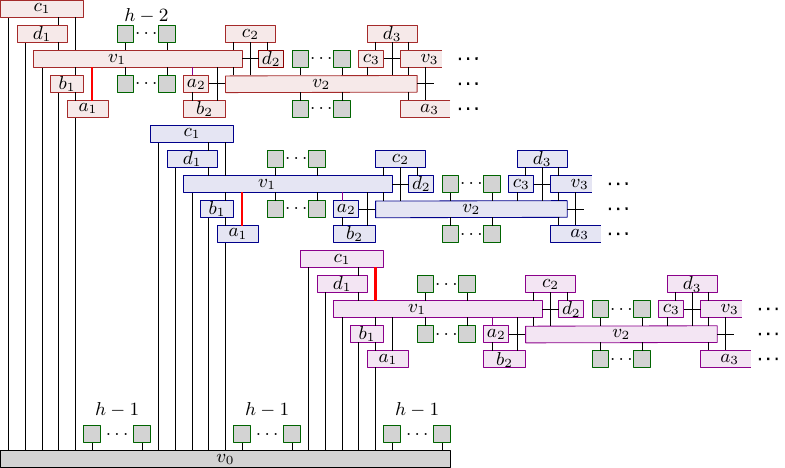}}
\subfigure{\includegraphics[width=0.65\textwidth,page=4]{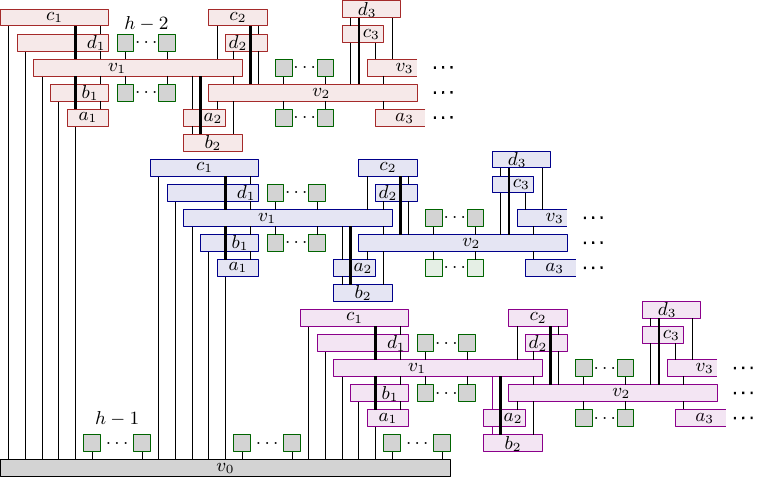}}
\subfigure{\includegraphics[width=0.65\textwidth,page=12]{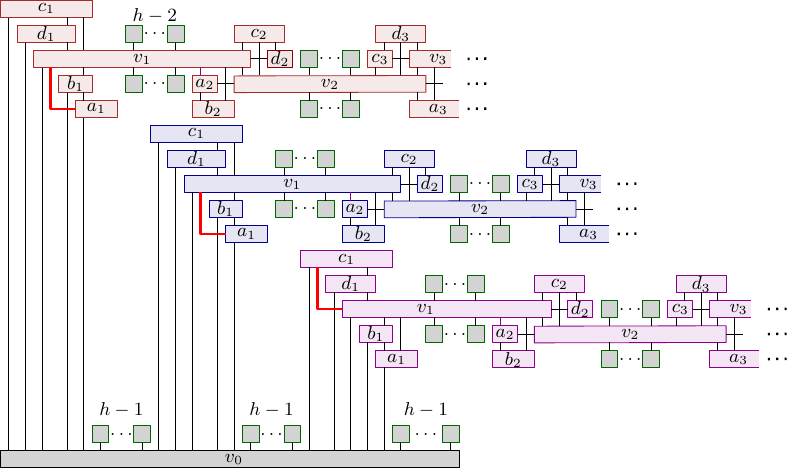}}
\caption{Drawings of $G_{h,l}$ with constant height, from top to bottom: 
an embedding-preserving ortho-polygon visibility representation with vertex complexity $1$, 
a bar-visibility representation for a changed embedding (see the red edges),
an embedding-preserving bar-1-visibility representation, 
and
an embedding-preserving 1-bend orthogonal box-drawing (the red edges have a bend). 
\label{fig:linear}}

\end{figure}
} 

\iffull
Both theorems are tight, as we will show below.
We remark that the graphs in the proof of Theorem~\ref{thm:lowerGeneral} admit drawings of height $\Theta(1)$ (hence area $\Theta(n)$) if we adopt different drawing paradigms; see Fig.~\ref{fig:linear}.
\fi

\section{Optimal Area Drawings}

\abstractlater{\section{Missing details from Section~\ref{sec:sqrt}}}
\label{sec:sqrt}

In this section we show how to compute an embedding-preserving visibility representation  of area $O(n^{1.5})$ which is tight by Thm.~\ref{thm:lowerGeneral}.   
By Obs.~\ref{obs:widthHeightUpperTrivial} it suffices to construct a drawing of height $O(\sqrt{n})$.
\onlyfull{Our construction is quite lengthy and outlined as follows.      After some preliminaries (Section~\ref{sec:sqrt_prelim}) we discuss how to choose a path $\pi$ in $\dual{G}$ (Section~\ref{sec:sqrt_pi}).  This path has properties that were crucial for so-called \emph{LR-drawings} of trees \cite{Chan02}; our drawing approach is inspired by doing LR-drawings of the inner dual, but of course we must draw the primal graph.  
Consider first how to draw a graph $P_\pi$ whose inner dual is a path (Section~\ref{sec:sqrt_path}).  This is easy if we assume that each ``hanging subgraph'' (i.e., a maximum subgraph in $G\setminus P_\pi$) has a drawing where the two vertices in common with $P_\pi$ are in a specific position.  Unfortunately  our path-drawing does not satisfy this condition.  So we change the approach and first draw a larger subgraph that includes $P_\pi$.  This subgraph consists of a 
``cap'' (Section~\ref{sec:sqrt_cap}), which we can draw easily to satisfy the position-requirement, and the ``handle'', for which we use the path-drawing.  By increasing the height proportionally to the degree of some vertex, we can combine the cap-drawing with the path-drawing (Section~\ref{sec:sqrt_umbrella}).
\todo{Renamed roof to cap and change $R$ to $C$.  Undoubtedly a few got overlooked; watch out for this during proofreading.}
To remove the dependency on the degree, we actually draw multiple caps before drawing the rest (Section~\ref{sec:sqrt_kumbrella}) and can then guarantee $O(\sqrt{n})$ height and hence $O(n^{1.5})$ area.
} 

\onlyabstract{Our construction is quite lengthy, so we mostly sketch it here via figures.
We assume that $G$ is maximal-planar and a \emph{reference-edge} $(s,t)$ on the outer-face of $G$ is fixed, and first choose a path $\pi$ in dual tree $\dual{G}$ (rooted at the face incident to $(s,t)$).  Let $F\in \Theta(n)$ be the size of $\dual{G}$ (hence the number of inner faces of $\skel{G}$).  As shown by Chan for binary trees \cite{Chan02} and generalized by us to arbitrary trees \cite{biedl2021generalized}, $\pi$ can be chosen such that $\alpha^p+\beta^p\leq (1-\delta)F^p$, where $\alpha$ [$\beta$] is the maximum size of a left [right] subtree of $\pi$, $p=0.48$ and $\delta>0$ is a constant.
Define a recursive function $h(F)=\max_{\alpha,\beta} \{h(\alpha)+h(\beta)\}+O(\sqrt{F})$
(with appropriate constants and base cases). Here the maximum is over all choices of $\alpha,\beta$ that satisfy the inequality.  
We construct a drawing of height $h(F)$.


So we first discuss how to draw the outer-1-path $P_\pi$ whose inner dual is $\pi$, 
plus all its \emph{hanging subgraphs} (i.e., maximum subgraphs in $G\setminus P_\pi$).  To do so, first create a visibility representation of $P_\pi$ on 5 rows such that edges with attached hanging subgraphs are drawn horizontally in the top or bottom row 
(see Fig.~\ref{fig:drawPathAbstract}).
Assume that each hanging subgraph $H$ has a 
\emph{$TC_{\sigma,\tau}$-drawing} (for $\{\sigma,\tau\}=\{1,2\}$), i.e., a drawing where the endpoints of the reference-edge occupy the $T$op $C$orners and have height $\sigma$ and $\tau$.  Then we can easily merge all hanging subgraphs, after expanding some boxes of $P_\pi$ one row outward.  The resulting drawing has height $h(\alpha)+h(\beta)+O(1)$ since all hanging subgraphs below [above] correspond to left [right] subtrees of $\pi$.\footnotemark

\begin{figure}[htp!]
\centering
{\includegraphics[page=6,scale=0.43,trim=0 0 100 0,clip]{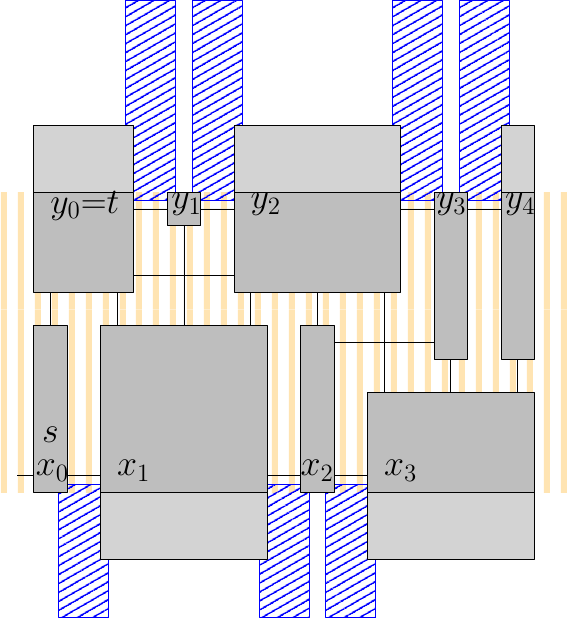}} \\
\caption{Drawing an outer-1-path (dark gray) and merging hanging subgraphs (blue striped) after expanding some vertex-boxes (light gray).}
\label{fig:drawPathAbstract}
\end{figure}

Alas, this path-drawing is not a $TC_{\sigma,\tau}$-drawing as required for the recursion.
So we change the approach and first draw a larger subgraph that includes $P_\pi$.   First extract the \emph{cap} $C_1$, consisting of all neighbours of $s$ and $t$, see Fig.~\ref{fig:2umbrellaAbstract}.
This is an outer-1-path, so we can use a path-drawing and get a $TC_{\sigma,\tau}$-drawing of the cap (see the corresponding part in Fig.~\ref{fig:draw2umbrellaAbstract}).  
The part of $P_\pi$ not in $C_1$  could be drawn as for outer-1-paths, but instead we first extract another cap $C_2$ at the edge common to $C_1$ and the rest of $P_\pi$.  We draw $C_2$ as a path and place it (after suitable expansion of the vertices of $C_1$) below the drawing of $C_1$.  This repeats $k$ times for some parameter $k$ of our choice ($k=2$ in the example). Then we draw the rest of $P_\pi$ (which we call \emph{handle}) as a path.

A major difficulty is combining the drawing of the caps with the handle-drawing.  Let $(x_i,y_j)$ be the edge common to caps and handle.  It is not too difficult to change the boxes of $x_i$ and $y_j$ to combine the two boxes that represented them in the two drawings, see also Fig.~\ref{fig:draw2umbrellaAbstract}.  The main challenge is that for the two hanging subgraphs incident to $y_j$, there is no suitable place to merge a $TC_{\sigma,\tau}$-drawing.
To resolve this, we split these hanging subgraphs further, and can then merge all their parts after adding $d(y_j)$ more rows, where $d(y_j)$ is the number of edges that $y_j$ has in these subgraphs (Fig.~\ref{fig:closeup}).  

So the goal is to choose the parameter $k$ such that $d(y_j)$ is small, because we need $d(y_j)$ additional rows 
beyond the $h(\alpha)+h(\beta)$ that we budget for hanging subgraphs.
Each extra cap also requires $O(1)$ additional rows 
but changes which vertex will take on the role of $y_j$.  
Crucially, the vertices $Y_1,\dots,Y_k$ that take on the role of $y_j$ have disjoint edge sets that count for $d(\cdot)$.  
Since there are $O(n)$ edges in total, there exists a $k\in O(\sqrt{n})$ such that $d(Y_k)\in O(\sqrt{n})$.  
With this choice of $k$, the recursive formula for the height hence becomes $h(\alpha)+h(\beta)+O(\sqrt{n})$, which by $F\in \Theta(n)$ and $\alpha^p+\beta^p\leq (1-\delta)n^p$ resolves to $O(\sqrt{n})$.  $(\star)$

\begin{figure}[htp!]
\centering
{\scalebox{1}[1]{\includegraphics[scale=0.35,page=2]{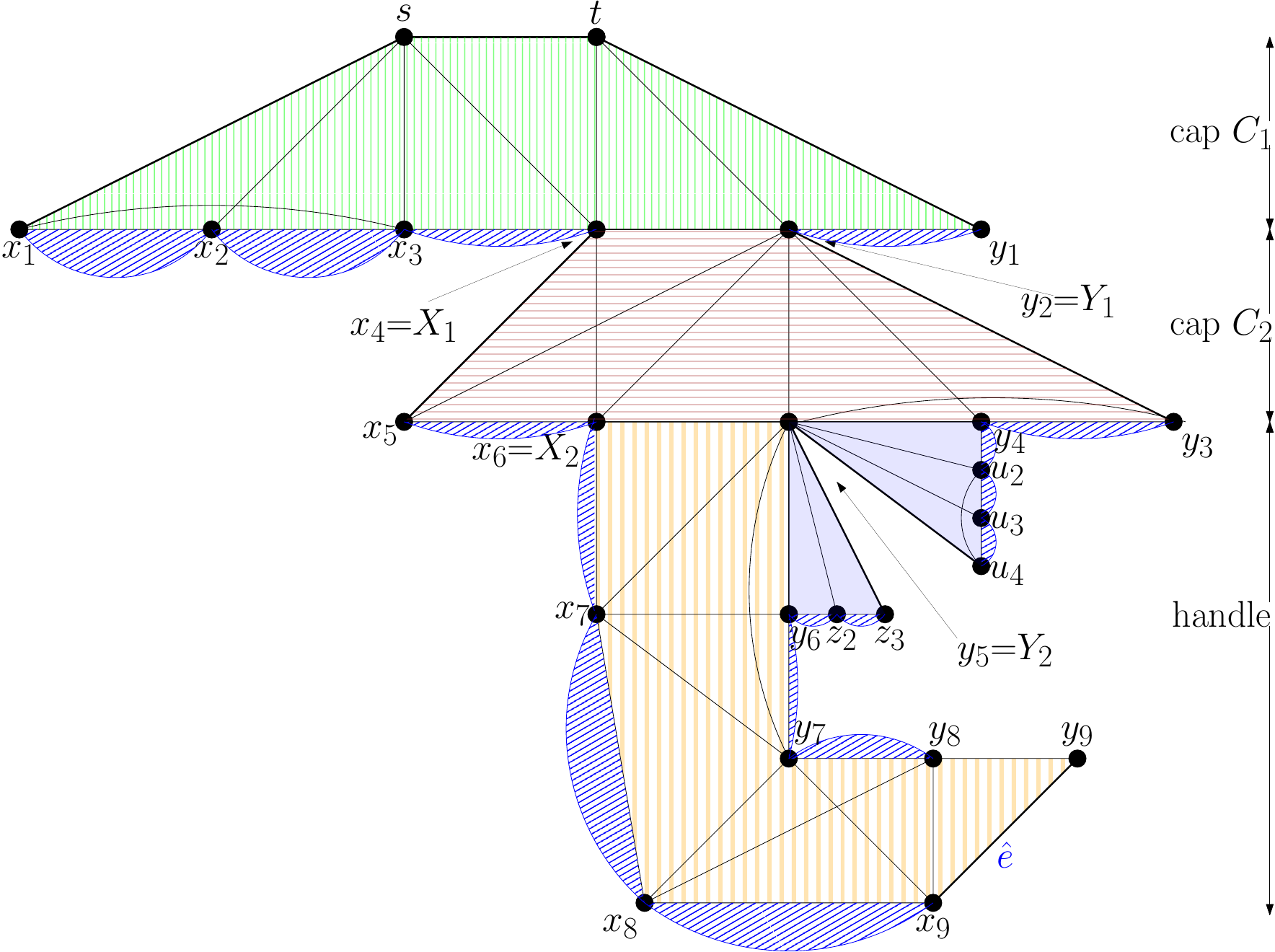}}}
\caption{Our running example.}
\label{fig:2umbrellaAbstract}
\end{figure}
} 

\later{
\subsection{Preliminaries}
\label{sec:sqrt_prelim}
We first need a few definitions and assumptions which will also be used in Section~\ref{sec:breaking} and~\ref{sec:bar}.
We shall assume without loss of generality that the input graph is plane-maximal: If not, we can always augment it with dummy edges which will be removed at the end of the construction. 

%
Let $G$ be a plane-maximal outer-$1$-plane graph. 
We assume throughout that a \emph{reference-edge} $(s,t)$ has been fixed, which is an edge on the outer-face, with $s$ before $t$ in the clockwise order of vertices along the outer-face.
Recall that $\dual{G}$ denotes the inner dual of the planar skeleton of $G$;
we consider $\dual{G}$ to be rooted at the face $f_{st}$ of $\skel{G}$ that is incident to $(s,t)$. A {\em root-to-leaf} path $\pi$ in $\dual{G}$ is a path that begins at $f_{st}$ and ends at a leaf of $\dual{G}$. 
Any maximal subtree of $\dual{G}\setminus V(\pi)$ is a rooted subtree that can 
be classified as \emph{left} or \emph{right subtree} depending on whether its 
root is left or right of the path-child at its parent
in $\pi$.

The \emph{$\pi$-path} $P_\pi$ of $G$ is the plane-maximal outer-1-path $P_\pi$ whose inner dual is $\pi$.   Since $\pi$ is a root-to-leaf path, $P_\pi$ contains $(s,t)$ on its outer-face, and we can pick one outer-face edge $\hat{e}\neq (s,t)$ at the face of $\skel{P_i}$ corresponding to the leaf of $\pi$.  We call these the \emph{end-edges} of $P_\pi$.

As outlined, we create our drawings by first drawing a subgraph $U$ that contains all of $P_\pi$  and then merging the rest. 
Enumerate the outer-face of $U$ in counter-clockwise direction as $s{=}x_0,x_1,\dots,x_\ell,y_r,y_{r-1},\dots,y_0{=}t$, where $(x_\ell,y_r)$ is the end-edge $\hat{e}$ of $P_\pi$.
(See e.g.~Fig.~\ref{fig:sqrt_umbrella}.)
We call $\langle x_0,x_1,\dots,x_\ell\rangle$ and $\langle y_0,y_1,\dots,y_r \rangle$ the {\em left} and \emph{right boundaries} (with respect to the end-edges). 
For $0\leq j<r$, edge $(y_j,y_{j+1})$ is on the outer-face of $U$, but need not be on the outer-face of $G$.  If it is not, then define the \emph{hanging subgraph} $H_{y_jy_{j+1}}$ of $U$ to be the subgraph induced by the vertices along the path on the outer-face between $y_j$ and $y_{j+1}$ that excludes $(s,t)$.  Since $U$ includes $P_\pi$, the inner dual of $H_{y_iy_{j+1}}$ is part of a right subtree of $\pi$; correspondingly we call $H_{y_jy_{j+1}}$ a \emph{right hanging subgraph}. The \emph{left hanging subgraph} $H_{x_ix_{i+1}}$ (for $0\leq i<\ell$) is defined symmetrically.

\subsection{Path $\pi$, function $h$ and drawing-types}
\label{sec:sqrt_pi}

We want a root-to-leaf path $\pi$ in $\dual{G}$ for which the sizes of the left
and right subtrees are balanced in some sense.  Chan \cite{Chan02} proved the 
existence of such paths for binary trees, and we recently generalized  this to arbitrary rooted trees:

\begin{lemma}[\cite{biedl2021generalized}]
\label{lem:pathTernary}
Let $p=0.48$.  Given any rooted tree $T$ of $n$ vertices, there exists a root-to-leaf
path $\pi$ such that for any left subtree $\alpha$ and any right subtree $\beta$
of $\pi$, $|\alpha|^p+|\beta|^p \leq (1-\delta)n^p$ for some constant
$\delta>0$.
\end{lemma}

Fix this path $\pi$
for the rest of this section.
Slightly abusing notation, we now use $\alpha$ and $\beta$ for the {\em size}
of the maximum left/right subtree of $\pi$ in $\dual{G}$ (rather than
the trees themselves).   Going over
to graph $G$, we will measure its size by the number $F$ of inner faces in $\skel{G}$.  Then $\alpha$ and $\beta$ are the maximum size of a left/right hanging subgraph of $P_\pi$.  Since we draw a super-graph $U$ of
$P_\pi$, $\alpha$ and $\beta$ are upper bounds on the size of a left/right hanging subgraph of $U$.

Let $h(F)$ be the recursive function that satisfies $h(1)=3$ and
$$h(F)=\max_{\alpha^p+\beta^p\leq (1-\delta)F^p} h(\alpha)+h(\beta)+11\sqrt{F}+7$$
for $F>1$, 
\footnote{We made no attempts to optimize the constants; they could like be improved with a more careful analysis.}
where $p=0.48$ and $\delta>0$ are as in Lemma~\ref{lem:pathTernary}.
One can easily show that $h(F)\leq c\sqrt{F}-7$,
where $c=\max\{10,12/\delta\}$.
%
It hence
suffices to find drawings of height $h(\alpha)+h(\beta)+11\sqrt{F}+7$ to
obtain drawings of height $O(\sqrt{F})=O(\sqrt{n})$
(observe that $F\in \Theta(n)$).

As outlined, we restrict drawings of hanging subgraphs as follows.
For integers $\sigma,\tau$,  let a 
\emph{$(\sigma,\tau)$-Top-Corner drawing} ($TC_{\sigma,\tau}$-drawing for short)
be a drawing of $G$ where $s$ and $t$ occupy the top corners, and the boxes of $s$ and $t$ have height $\sigma$ and $\tau$, respectively.
See Fig.~\ref{fig:drawingTypes}.  

For later merging-steps, we briefly mention here that these drawings can be modified to satisfy other properties.
Assume we have a $TC_{2,1}$-drawing.
Since $(s,t)$ is an edge, it must necessarily be drawn horizontally between the boxes of $s$ and $t$. Since the box of $t$ has height 1, no other vertex occupies the top row.  Without changing the height, we can therefore change $t$ into a bar that spans the entire top row and retract $s$ to be a bar on the left end of the second row; see Fig.~\ref{fig:drawingTypes}.    We call the result a \emph{$(2,1)$-Top-Bar drawing} ($TB_{2,1}$-drawing for short).
We can also transform the drawing into a $TC_{\sigma,\tau}$-drawing, for any $\tau\geq 1$ and $\sigma\geq \tau+1$, by inserting $\tau-1$ and $\sigma-\tau-1$ rows above and below the top row, respectively, and extending the vertex-boxes.
Similar transformations can be applied to a $TC_{1,2}$-drawing.

With this, we can state our overall goal:

\begin{figure}[htp!]
\centering
\subfigure{\includegraphics[page=1,scale=0.9]{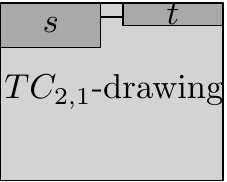}}\hfill
\subfigure{\includegraphics[page=2,scale=0.9]{types}}\hfill
\subfigure{\includegraphics[page=6,scale=0.9]{types}}\hfill
\subfigure{\includegraphics[page=4,scale=0.9]{types}}\hfill
\subfigure{\includegraphics[page=5,scale=0.9]{types}}\hfill
\caption{Drawing-types (the rightmost two are needed in Section~\ref{sec:breaking}).}
\label{fig:drawingTypes}
\end{figure}

\begin{lemma}
\label{lem:sqrt}
Let $G$ be an outer-1-plane graph with reference-edge $(s,t)$.
Then for any $\{\sigma,\tau\}\in \{1,2\}$, $G$ has an embedding-preserving visibility representation that is a $TC_{\sigma,\tau}$-drawing with height at most
$h(F)$. 
\end{lemma}

The proof of this lemma is by induction on $F$. Since $G$ is plane-maximal, if $F=1$ then $G$ is either a triangle or a $K_4$ drawn with one crossing.  Either way we can easily find a $TC_{\sigma,\tau}$-drawing of height at most $3=h(1)$.
The remainder of this section will prove the induction step.
Fix $\pi,\alpha,\beta$ as explained above.  

\subsection{Drawing a path. }
\label{sec:sqrt_path}

As a first step, we will give a simple construction that unfortunately
does not place $s,t$ where we need them, but places them on the left
side instead.  

\begin{lemma}
\label{lem:drawPath}
Fix $\{\sigma',\tau'\}=\{2,3\}$.
$G$ has an embedding-preserving
visibility representation $\Gamma$ of height $h(\alpha)+h(\beta)+3$.
Furthermore,  $P(s)$ and $P(t)$ have height $\sigma'$ and $\tau'$, respectively, both abut the left side of the bounding box, and there are $h(\alpha)-1$ rows below $P(s)$ and $h(\beta)-1$ rows above $P(t)$.
\end{lemma}
\begin{figure}[htp!]
\centering
\includegraphics[page=6,scale=0.5]{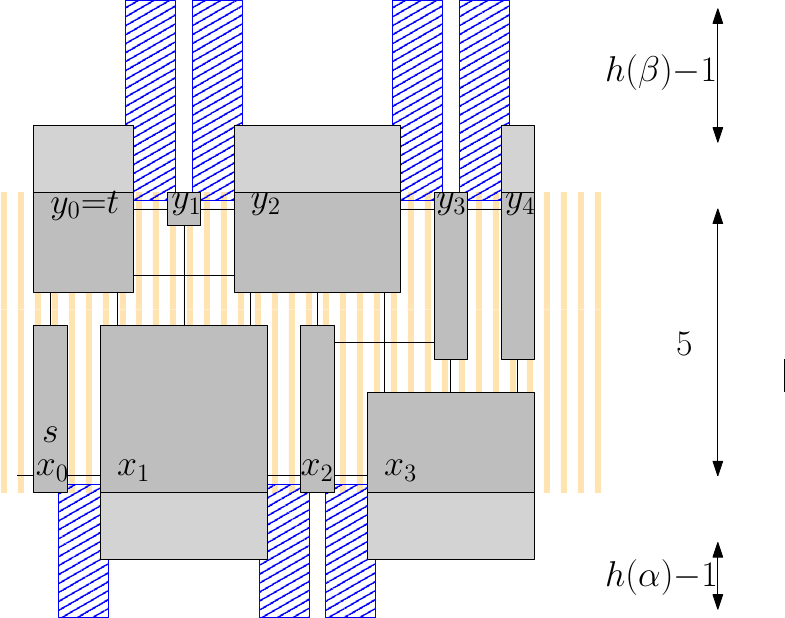}
\caption{Drawing an outer-1-path (dark gray) and merging hanging subgraphs (blue striped).}
\label{fig:drawPath}
\end{figure}
\begin{proof}
We first draw $P_\pi$ on five rows 
(see the orange hatched part of Fig.~\ref{fig:drawPath})
and then merge the hanging subgraphs.   

Assume the outer-face of $P_\pi$ is enumerated as $x_0,\dots,x_\ell,y_r,\dots,y_0$ as before.
The boxes of $x_0,\dots,x_\ell$ all intersect row 1 (of our five rows), in
this order from left-to-right, while the boxes of $y_0,\dots,y_r$ all
intersect row 5.  The outer-face edges along these paths can hence be 
realized horizontally.
All vertex-boxes have height 1, 2 or 3, and we determine
this (as well as their $y$-coordinates) by parsing the faces $f_1,\dots,f_k$
of $\skel{P_\pi}$, beginning at $f_{st}$, and extending the drawing rightwards.
The heights of $x_0{=}s$ and $y_0{=}t$ are determined
by $\sigma'$ and $\tau'$.  Assume we have drawn everything up to some face $f_{h-1}$,
which ends at edge $(x_i,y_j)$. Boxes $P(x_i),P(y_j)$ are partially drawn and their heights are $\{2,3\}$ (but we do not know which one has which height).  
The next face $f_h$ can have one of five configurations: 
\begin{itemize}
\item $f_h$ has no crossing (hence it is a triangle).  One of $\{x_i,y_j\}$ also
	belongs to $f_{h+1}$ (or is in $\{x_\ell,y_r\}$ if $h=k$). Let us assume that
	this is $y_j$, the other case is symmetric.
	See e.g.~$\{x_0,x_1,y_0\}$ in Fig.~\ref{fig:drawPath}.
	In this case, $P(x_i)$ ends, $P(x_{i+1})$ begins (with the same height as $P(x_i)$)
	and $P(y_j)$ extends further rightwards.  Edge $(x_{i+1},y_j)$ can be inserted
	vertically.
\item $f_h$ has a crossing and neither $x_i$ nor $y_j$ belongs to $f_{h+1}$ 
	(or is in $\{x_\ell,y_r\}$ if $h=k$).   
	See e.g.~$\{x_3,x_4,y_2,y_3\}$.  We call such a crossing an \emph{opposite-boundary crossing} since both crossing edges connect opposite boundaries.  Let us assume that $P(x_i)$ has height 3 and $P(y_j)$ has height 2; the other case is symmetric.  In this case, $P(x_i)$ ends and $P(x_{i+1})$ begins with height 2, which means that we can draw $(x_{i+1},y_{j})$ vertically.  Then $P(y_j)$ ends and $P(y_{j+1})$ begins with height 3, which means that we can draw $(x_i,y_{j+1})$ horizontally along row 3.
\item $f_h$ has a crossing and one of $\{x_i,y_j\}$ belongs to $f_{h+1}$ 
	(or is in $\{x_\ell,y_r\}$ if $h=k$).   
	Let us assume that this is $x_i$, the other
	case is symmetric.  See e.g.~$\{x_1,y_0,y_1,y_2\}$.  
	   We call such a crossing a \emph{same-boundary crossing} since one of the crossing edges connects
	    two vertices on the same boundary.
	In this case, $P(y_j)$ ends and $P(y_{j+1})$ begins (with height one less than the height of $P(y_j)$), which means that we can draw $(y_{j+1},x_i)$ vertically.  Then $P(y_{j+1})$ ends and $P(y_{j+2})$ begins (with the same height as $P(y_j)$), which means that we can draw $(y_j,y_{j+2})$ horizontally along the bottom row of $P(y_j)$ and $P(y_{j+2})$.
\end{itemize}

Extending the drawing one face at a time gives us an embedding-preserving visibility representation $\Gamma_P$ of $P_\pi$.
Before merging hanging subgraphs, we first parse along the top and bottom row from left to right and
extend every second box by one row ``outward'' (i.e., away from $\Gamma_P$).
Consider a left hanging subgraph $H_{x_ix_{i+1}}$.   Edge $(x_i,x_{i+1})$
is drawn horizontally along the bottom row of $\Gamma_P$. After the extension,
one of $x_i,x_{i+1}$ (say $x_i$) extends one row further down.  This means that we
can merge a (recursively obtained) $TC_{2,1}$-drawing $\Gamma_{H}$
of $H_{x_ix_{i+1}}$  in the
rows below $(x_i,x_{i+1})$.%
\footnote{In all our figures, we assume that the subgraph-drawings have been
scaled horizontally so that they fit.  Put differently, we do not assume that
$x$-coordinates are integers; they can be made integers by inserting columns as needed.  Also, hanging subgraphs do not necessarily all have to exist or have the same height, we show here the maximum height that could be needed.}
This requires at most $h(\alpha)-1$ rows below the 
rows for $P_\pi$ since $\Gamma_{H}$ has height at most $h(\alpha)$
and the row for edge $(x_i,x_{i+1})$ can be used by both $\Gamma_P$ and $\Gamma_{H}$.
Symmetrically we can merge right hanging subgraphs using up to $h(\beta)-1$
rows above $\Gamma_P$.\footnotemark
\end{proof}
} 

\footnotetext{Readers familiar with LR-drawings \cite{Chan02,FPR20,biedl2021generalized} may notice the similarity of constructing the path-drawing with the (rotated) LR-drawing of $\pi$, except  that we draw the outer-1-planar graph rather than its dual tree.}
\later{
\subsection{Drawing a cap}
\label{sec:sqrt_cap}

Define the {\em cap} $C$ to be the outer-1-path that contains $s,t$ and all vertices adjacent to $s$ or $t$, and let the
\emph{umbrella} \cite{BDa-WADS17} be the union of $C$ and $P_\pi$; we denote it by $U_\pi$.  
See also Fig.~\ref{fig:sqrt_umbrella}.
We will later draw all of $U_\pi$, but for now are only concerned with drawing $C$.  Enumerate the outer-face of $U_\pi$ as $x_0{=}s,\dots,x_r,y_\ell,\dots,y_0{=}t$.  Assume that $P_\pi\subsetneq C$ (in the other case we \emph{only} have to draw $C$, which will be easier). Of special interest is then the \emph{transition edge} $(x_i,y_j)$,
which is the edge on the outer-face of the cap that is an inner edge of $P_\pi$.  Sometimes we use the notation $X:=x_i$ and $Y:=y_j$.  The \emph{handle}
$P_h$ is the part of $P_\pi$ not in $C$; thus $P_h$ is the outer-1-path that contains all vertices of $P_\pi$ between $(X,Y)$ and end-edge $\hat{e}$.    
Note that the handle (plus its hanging subgraphs) can be viewed as the hanging subgraph $H_{XY}$ of the cap.
\begin{figure}[htp!] 
\subfigure[\label{fig:sqrt_umbrella}]{\scalebox{1}[0.8]{\includegraphics[scale=0.25,page=1,trim=240 0 90 0,clip]{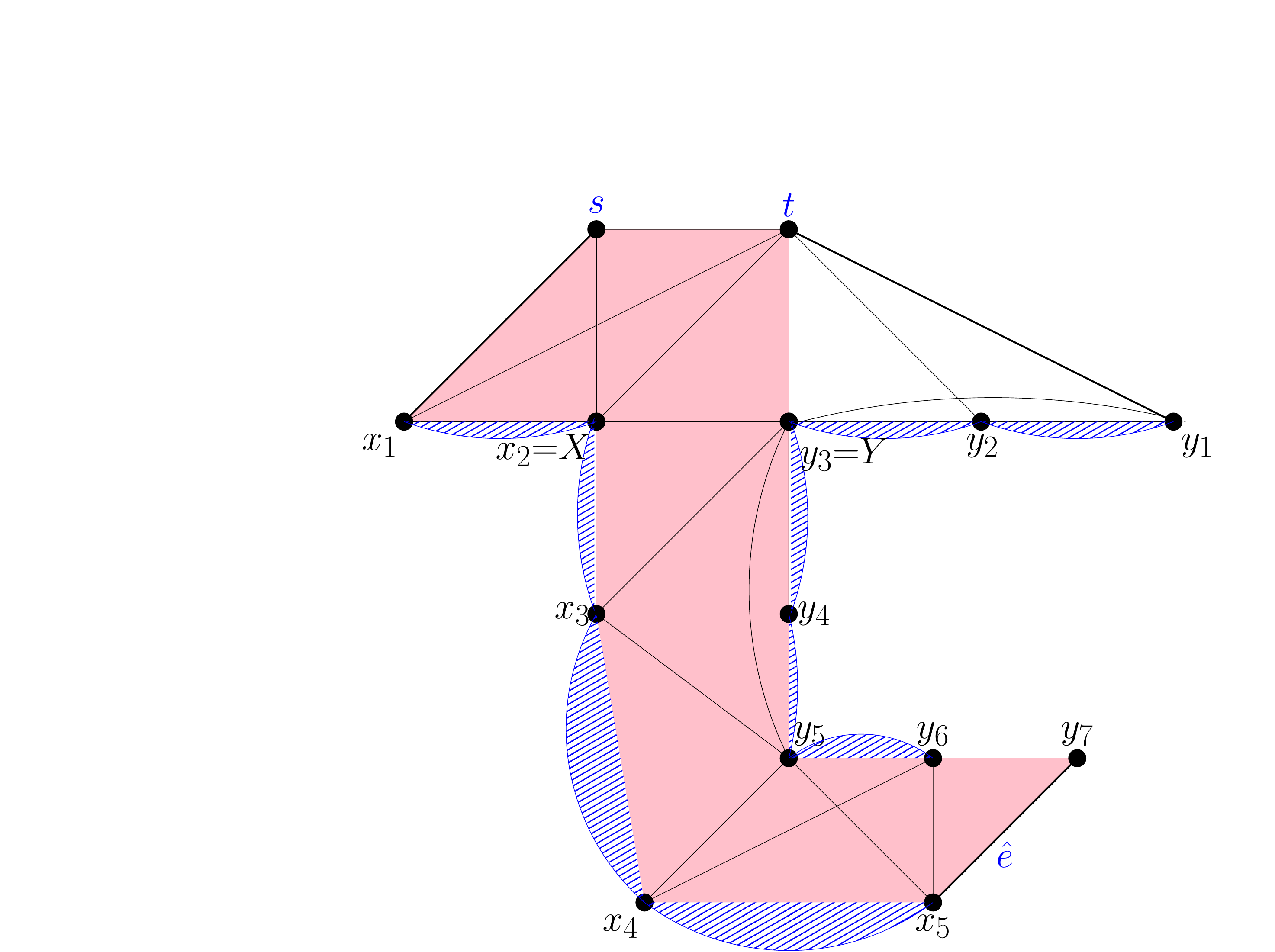}}}\hspace*{-5mm}
\subfigure[\label{fig:2umbrella}]{\scalebox{1}[0.8]{\includegraphics[scale=0.25,page=2]{example.pdf}}}
\caption{(a) An umbrella $U_\pi$.  $P_\pi$ is shaded.  (b) Our running example with its 2-cap umbrella $U^{(2)}$ and the hanging subgraphs $H_{y_5y_4}$ and $H_{y_6y_5}$ expanded.     Fig.~\ref{fig:drawPath} was a drawing of its handle.
\todo[inline]{Time permitting, redraw this figure to have smaller horizontal spacing and perhaps a different layout at the ``wings'' so that we can make labels bigger.}
}
\end{figure}

\begin{figure}[htp!]
\centering
\includegraphics[page=7,scale=0.4]{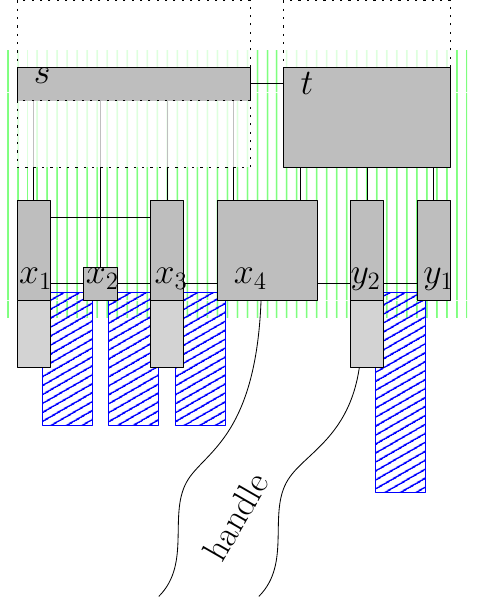}
\hfil
\includegraphics[page=8,scale=0.4]{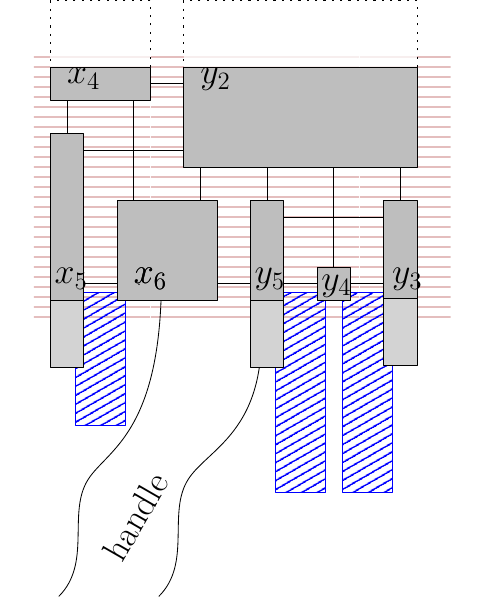}
\caption{A $TC_{1,2}$-drawing of a cap, both when $f_{st}$ has a crossing and otherwise.  The naming is consistent with the example from Fig.~\ref{fig:2umbrella}. \label{fig:draw_cap}  
}
\end{figure}

\begin{lemma}
\label{lem:draw_cap}
Let $C$ be the cap, and let $\{\sigma,\tau\}=\{1,2\}$.  The subgraph
formed by $C$ and its hanging subgraphs (except $H_{XY}$) has an
embedding-preserving visibility representation of type $TC_{\sigma,\tau}$ and 
height $\max\{h(\alpha),h(\beta)\}+3$.
\end{lemma}
\begin{proof}
The path-drawing from Lemma~\ref{lem:drawPath} satisfies this, as long as we
re-define the path that we use and modify boxes a bit.
Note the cap can be viewed as $P_{\pi'}$ for the path $\pi'$ that begins at the inner
face of $\skel{G}$ incident to $(x_1,s)$ and ends at the inner face of
$\skel{G}$ incident to $(t,y_1)$.
Choose $\tau'$ as follows:  If $f_{st}$ has a crossing, then set $\tau'=\sigma+1$, otherwise set $\tau'=3$.
Apply Lemma~\ref{lem:drawPath} with respect to this path $\pi'$,   reference-edge $(x_1,s)$, and values $\tau'$ and $\sigma'=\{2,3\}\setminus \tau'$
hence $P(x_1)$ has height $\sigma'$ and $P(s)$ has height $\tau'$.
Since $(s,t)$ is on the outer-face, there are no hanging subgraphs to be
merged above, and the top two rows of the drawing contain
only $s$, $t$, and edge $(s,t)$,     we can delete one of these rows.   If $f_{st}$ had no crossing, then both $P(s)$ and $P(t)$ now have height 2; reduce the height of one of them to 1 (after extending vertical edges) as dictated by $\{\sigma,\tau\}$.  If $f_{st}$ had a crossing, then $P(s)$ now has height $\sigma$ and the height of $P(t)$ is different from $P(s)$, hence $P(t)$ has height $\tau$.
All
other conditions are easily verified.
\end{proof}

\subsection{Drawing an umbrella.}
\label{sec:sqrt_umbrella}

Now we need to combine the cap-drawing of Lemma~\ref{lem:draw_cap} with the path-drawing of Lemma~\ref{lem:drawPath} (applied to the handle).
This combination step is
the most challenging part of our construction and illustrated in
Fig.~\ref{fig:drawUmbrella}. 
Assume $\{\sigma,\tau\}=\{1,2\}$ has been given.  
We split the explanation of how to draw $G$ into two steps.

\paragraph{Step 1: Draw cap, handle, and most hanging subgraphs.}
We have two subcases, depending on the configuration of $(x_i,y_j)$ in handle $P_h$.
Let $f,f'$ be the first two faces of $\skel{P_h}$, enumerated starting
at transition edge $(x_i,y_j)$.
Not both $x_i$ and $y_{j}$ can belong to $f'$.  We assume here that $x_i$
does {\em not} belong to $f'$, the other case is symmetric (we would
merge the handle in leftward direction).

Using Lemma~\ref{lem:draw_cap},
obtain a $TC_{\sigma,\tau}$-drawing $\Gamma^+_C$ of cap $C$ and its hanging subgraphs except $H_{XY}$.
We remove the drawings of hanging subgraphs $H_{x_{i-1}x_i}$ and $H_{y_{j-1}y_{j}}$ from $\Gamma_C^+$; they will be handled later. 
Recall that as part of creating $\Gamma_C^+$ we expanded every second box of $x_1,\dots,x_i,y_j,\dots,y_1$ downward by one unit; we assume here that the choice has been done such that $x_i$ is extended downward and $y_j$ is not.

Using Lemma~\ref{lem:drawPath}, obtain a drawing $\Gamma_P^+$ of $P_h$ and its hanging subgraphs (i.e., graph $H_{XY}$), using $(x_i,y_j)$ as reference-edge and $\sigma'=3$ so that $P(x_i)$ has height 3.  We use $\Gamma_P$ to denote the drawing of $P_h$ within $\Gamma_P^+$.  
Omit from $\Gamma_P^+$ the drawings of hanging subgraphs  $H_{x_ix_{i+1}}$ and $H_{y_jy_{j+1}}$ for special handling later. 

To merge the two drawings, insert sufficiently many new columns between $y_j$ and $y_{j-1}$ in $\Gamma_C^+$, and place $\Gamma_P^+$ sufficiently far below the horizontal line segment $(y_j,y_{j-1})$.  Here, ``sufficiently far'' means that there are at least $h(\beta)-1$ rows between $(y_j,y_{j-1})$ and $\Gamma_P$; therefore the right hanging subgraphs of $P_h$ fit without overlapping $\Gamma_C^+$.  (We will actually need even more rows later.)


We now let $P(y_{j})$ be the minimum box that includes both  of boxes of $y_j$ in $\Gamma^+_C$ and $\Gamma^+_P$; with our placement of $\Gamma^+_P$ and because we removed hanging subgraphs this does not overlap  other vertices.    
To unify the two copies of $x_i$, we need to be more careful.  Expand the copy of $x_i$ in $\Gamma^+_C$ vertically downward until the rows containing the copy of $x_i$ in $\Gamma^+_P$; this is $P(x_i)$.  Delete the copy in $\Gamma^+_P$, and expand all its incident horizontal edges leftward until $P(x_i)$.  By case-assumption $x_i$ belongs only to the first face of $P_h$ and has height 3.  
One verifies that therefore (in $\Gamma_P^+$) all its incident vertically drawn edges are either $(x_i,y_j)$ (which we can ignore because it is also realized in $\Gamma^+_C$) or belong to the hanging subgraph $H_{x_ix_{i+1}}$ (whose drawing was omitted).  So extending only horizontal edges suffices to draw all incident edges at $x_i$. 

\begin{figure}[htp!]
\centering
\includegraphics[scale=0.39,page=9]{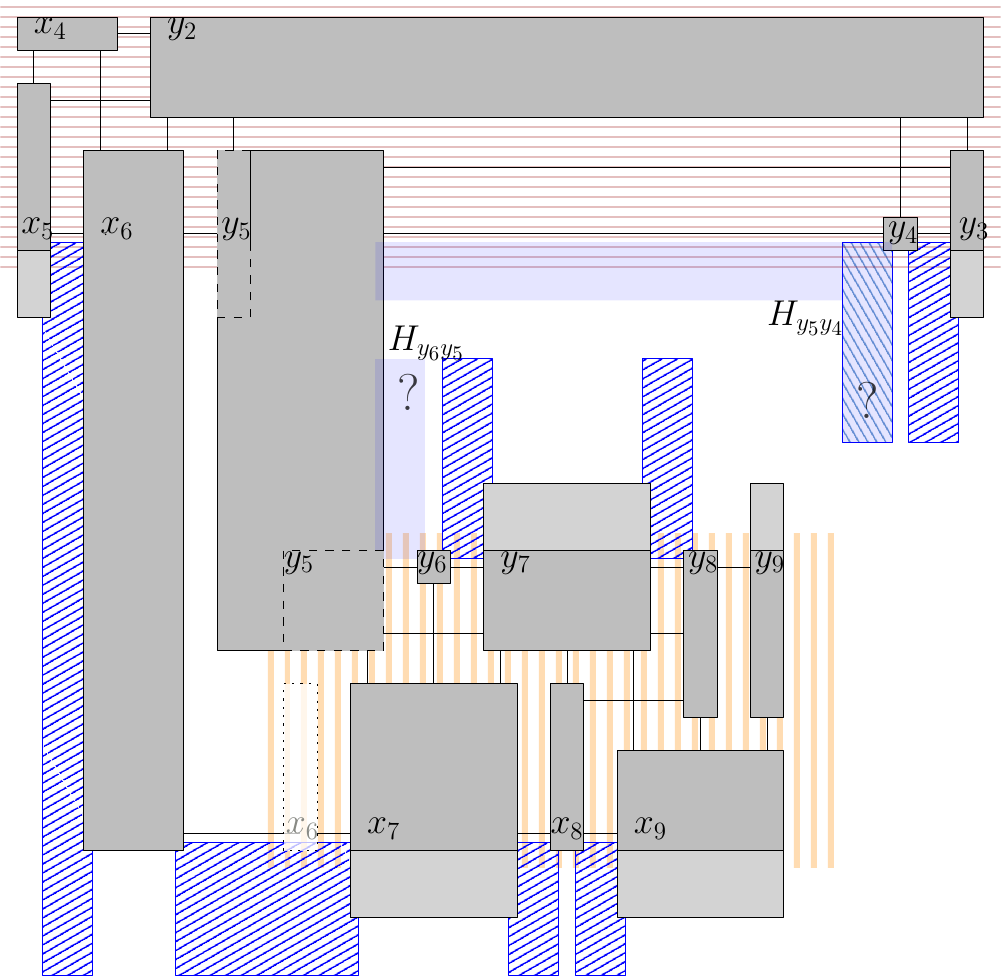}
\caption{Drawing an umbrella.  Two hanging subgraphs are not yet included.
}
\label{fig:drawUmbrella}
\end{figure}


We have four hanging subgraphs that were omitted.  Two of them are easily merged as follows.  For graph $H_{x_{i-1}x_i}$, use a $TC_{1,h}$-drawing, where $h\geq 2$ is the height of $P(x_i)$; recall that we can create this from a $T_{1,2}$-drawing by inserting rows.
For graph $H_{x_ix_{i+1}}$, the bottom sides of $P(x_i)$ and $P(x_{i+1})$ are one row apart and so we can merge a $TC_{1,2}$-drawing almost as it would have been done in
$\Gamma^+_P$ (the only change is that we stretch it horizontally
so that it uses the new location of $P(x_i)$).

\paragraph{Step 2: Drawing hanging subgraphs at $y_j$.}

We still need to merge hanging subgraphs $H_{y_{j}y_{j-1}}$ and
$H_{y_{j+1}y_{j}}$, and this
causes major difficulties because the endpoints of their reference-edge
are not in a position that we can match in a drawing of the subgraph.
We solve this by breaking $H_{y_{j}y_{j-1}}$ and $H_{y_{j+11}y_j}$ down
further, and using one extra row for each edge incident to $y_{j}$ in these subgraphs.  Specifically,
let $u_1,\dots,u_d$ be neighbours of $y_{j}$ in $H_{y_{j}y_{j-1}}$ indexed in clockwise order, beginning  with $y_{j-1}{=}u_1$.
Insert $d$ new columns to the right of $\Gamma_P^+$ to widen $y_{j-1}$, and place $u_2,\dots,u_{d}$ as left-aligned bars of length $d{-}1,\dots,1$ in the newly added columns and below $u_1$.  We can connect them horizontally to $y_j$.
For any edge $(u_h,u_{h+2})$ that may exist, 
shorten the left end of $P(u_{h+1})$ by one unit so that a vertical visibility between $u_h$ and $u_{h+2}$ is formed. 
Assuming we left at least $d+h(\beta)-2$ rows between $(y_j,y_{j-1})$ and $\Gamma_P$, we can now merge the hanging subgraphs $H_{u_hu_{h+1}}$ in the rows below the staircase formed by the bars, using an $TB_{2,1}$-drawing. Note that crucially there is no hanging subgraph at $(y_{j},u_d)$, otherwise $y_j$ would have further neighbours in $H_{y_jy_{j-1}}$. 
So this merges all of $H_{y_{j}y_{j-1}}$ as required.

Similarly, to merge $H_{y_{i+1}y_i}$, let $z_1,\dots,z_{d'}$ be neighbours of $y_j$ in $H_{y_jy_{j+1}}$, indexed in counter-clockwise order, beginning at $y_{j+1}$.  Place $z_2,\dots,z_{d'}$ as bars forming a staircase, in the $d'{-}1$ rows above $y_{i+1}$. Shorten left ends if needed so that edges among them drawn.  Merge hanging subgraph $H_{z_hz_{h+1}}$ (for $h=1,\dots,d'-1$) in the rows above, and note that this merges all of $H_{y_{j+1}y_j}$ and creates no overlap as long as there are at least $d{-}1+d'{-}1+h(\beta){-}1=d+d'+h(\beta)-3$ rows between $(y_j,y_{j-1})$ and $\Gamma_P$.  See Fig.~\ref{fig:VRDelta}.
} 


\later{
\begin{figure}[htp]
\hspace*{\fill}
\includegraphics[scale=0.5,page=3,trim=250 180 110 180,clip]{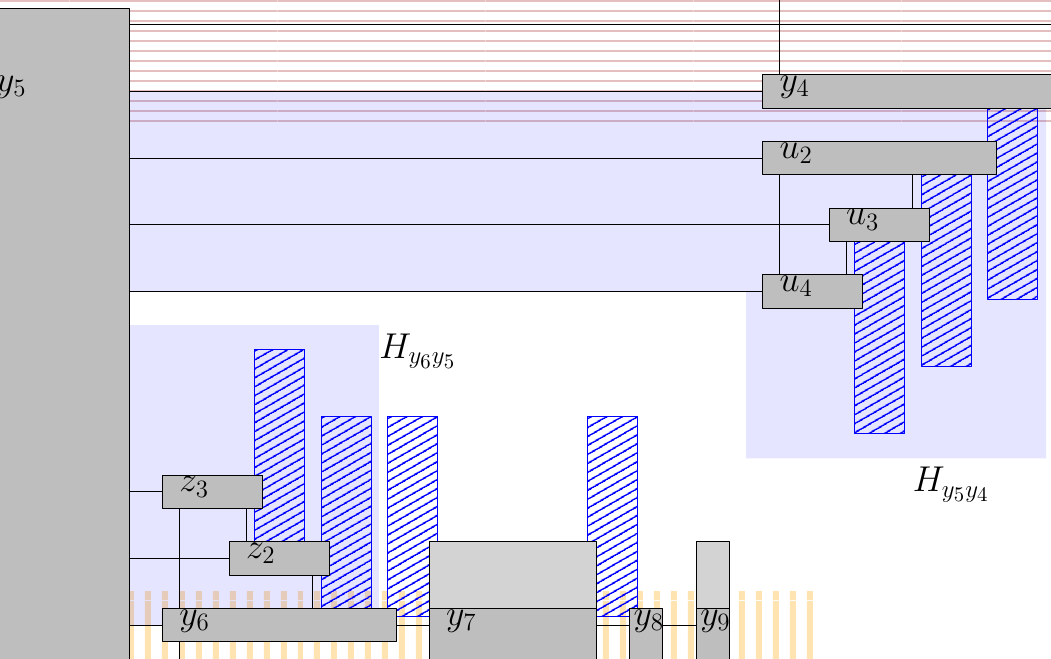}
\hspace*{\fill}
\caption{Closeup: Inserting the two remaining hanging subgraphs.
}
\label{fig:VRDelta}
\end{figure}

For any
vertex $v$ in $U_\pi$, let $d(v)$ be the number of edges 
incident to $v$ that do {\em not} belong to $U_\pi$.
The height of our construction depends on $d(y_j)$ as follows.

\begin{lemma}
\label{lem:drawUmbrella}
For any $\{\sigma,\tau\}=\{1,2\}$, $G$ has an embedding-preserving 
visibility representation $\Gamma$ 
that is a $TC_{\sigma,\tau}$-drawing
with height $h(\alpha)+h(\beta)+D+7$, where 
$D=\max\{d(x_i),d(y_j)\}$  if the umbrella has transition edge $(x_i,y_j)$ and  $D=0$ otherwise.
\end{lemma}
\begin{proof}
This holds by Lemma~\ref{lem:draw_cap} if there is no transition edge (and hence $H_{XY}$ is empty), so assume $(x_i,y_j)$ exists.
The drawing construction was already given above, so we only need to analyze
the height.  The left hanging subgraphs use at most
$h(\alpha){-}1$ rows below $\Gamma_P$.  
The right hanging subgraphs can be merged as long as we leave $d+d'+h(\beta)-3$ rows between $(y_j,y_{j-1})$ and $\Gamma_P$.
By $d(y_j)=d+d'-2$ this equals $d(y_j)+h(\beta)-1$ rows.
Recall that we made an assumption on $y_j$; in the
other symmetric case where we merge leftward we would use $x_i$ rather than $y_j$ here. 
Finally we need 9 rows for the boxes of the cap and the handle, 
and so the total height is as required.
\end{proof}

\subsection{Drawing a $k$-cap umbrella. }
\label{sec:sqrt_kumbrella}

If our graph has small maximum degree, then Lemma~\ref{lem:drawUmbrella}
gives us the required recursion and we are done.  We now show that even for
larger maximum degree we can achieve height $O(\sqrt{n})$ by extracting
multiple caps before drawing the remaining part of $P_\pi$.

Roughly speaking, a $k$-cap umbrella consists of taking caps $k$ times (always along the path $\pi$) and using the rest as handle.  Formally, 
let $C_1$ be the cap with respect to reference-edge $(s,t)$.
If all of $P_\pi$ belongs to $C_1$ then the recursive procedure stops, $(X_k,Y_k)$ is undefined and $C_2,\dots,C_k$ are the empty set.
Otherwise, let $(X_1,Y_1)$ be the transition edge of $C_1$
and let $H_1$
be the hanging subgraph at $(X_1,Y_1)$.  Repeat in $H_1$, i.e., let $C_2$ be the cap of $H_1$ (with respect to reference-edge $(X_1,Y_1)$), let $(X_2,Y_2)$ be its transition edge, let $H_2$ be the hanging subgraph at $(X_2,Y_2)$ etc., until we obtain cap $C_k$ and its transition edge $(X_k,Y_k)$. 
(If all of $P_\pi$ belongs to $C_1\cup \dots \cup C_i$ for some $i\leq k$ then $C_{i+1},\dots, C_k$ are empty sets and $(X_k,Y_k)$ is undefined.)  The \emph{$k$-cap umbrella} $U^{(k)}$ is then $C_1\cup \dots \cup C_k\cup P_\pi$, and its \emph{last transition edge} is $(X_k,Y_k)$ (which may be undefined).
See Fig.~\ref{fig:2umbrella}.

\begin{lemma}
\label{lem:drawkUmbrella}
For any $\sigma,\tau=\{1,2\}$,  and any $k\geq 1$,
$G$ has an embedding-preserving
visibility representation $\Gamma$ that is a $TC_{\sigma,\tau}$-drawing
with height $h(\alpha)+h(\beta)+3k+D+4$, where 
$D=\max\{d(X_k),d(Y_k)\}$ if the last transition edge is $(X_k,Y_k)$ and $D=0$ if the last transition edge is undefined.
\end{lemma}
\begin{proof}
For $k=1$ this holds by Lemma~\ref{lem:draw_cap}.
So assume $k>1$, and let $C_1,(X_1,Y_1),H_1$ be as above.
Using Lemma~\ref{lem:draw_cap},
create a drawing $\Gamma^+_C$ of cap $C_1$ and its hanging subgraphs (except $H_1$).
We are done if $P_\pi\subseteq C$, so assume not.
Recursively obtain a drawing $\Gamma_H$ of $H_1$ with respect to reference-edge $(X_1,Y_1)$ and parameter $k-1$, i.e., using the $(k-1)$-cap umbrella $U^{(k-1})$ of $H_1$.
By choosing a $C_{2,1}$-drawing or $C_{1,2}$-drawing as required by $\Gamma_C^+$, drawing $\Gamma_H$ can be merged easily below $(X_1,Y_1)$.
See Fig.~\ref{fig:draw2umbrella}.  The last transition edge $(X_k,Y_k)$ of $U^{(k)}$ (in $G$) is the same as the last transition edge of $U^{(k-1)}$ (in $H_1$),
so $\Gamma_H$ has height at most $h(\alpha)+h(\beta)+3(k-1)+D+4$.
Drawing $\Gamma_C^+$ uses three additional rows above $\Gamma_H$, 
and at most $3+\max\{h(\alpha),h(\beta)\}$ rows everywhere else.
The height-bound follows.
\end{proof}

\begin{figure}[htp!]
\centering
\includegraphics[scale=0.39,page=4]{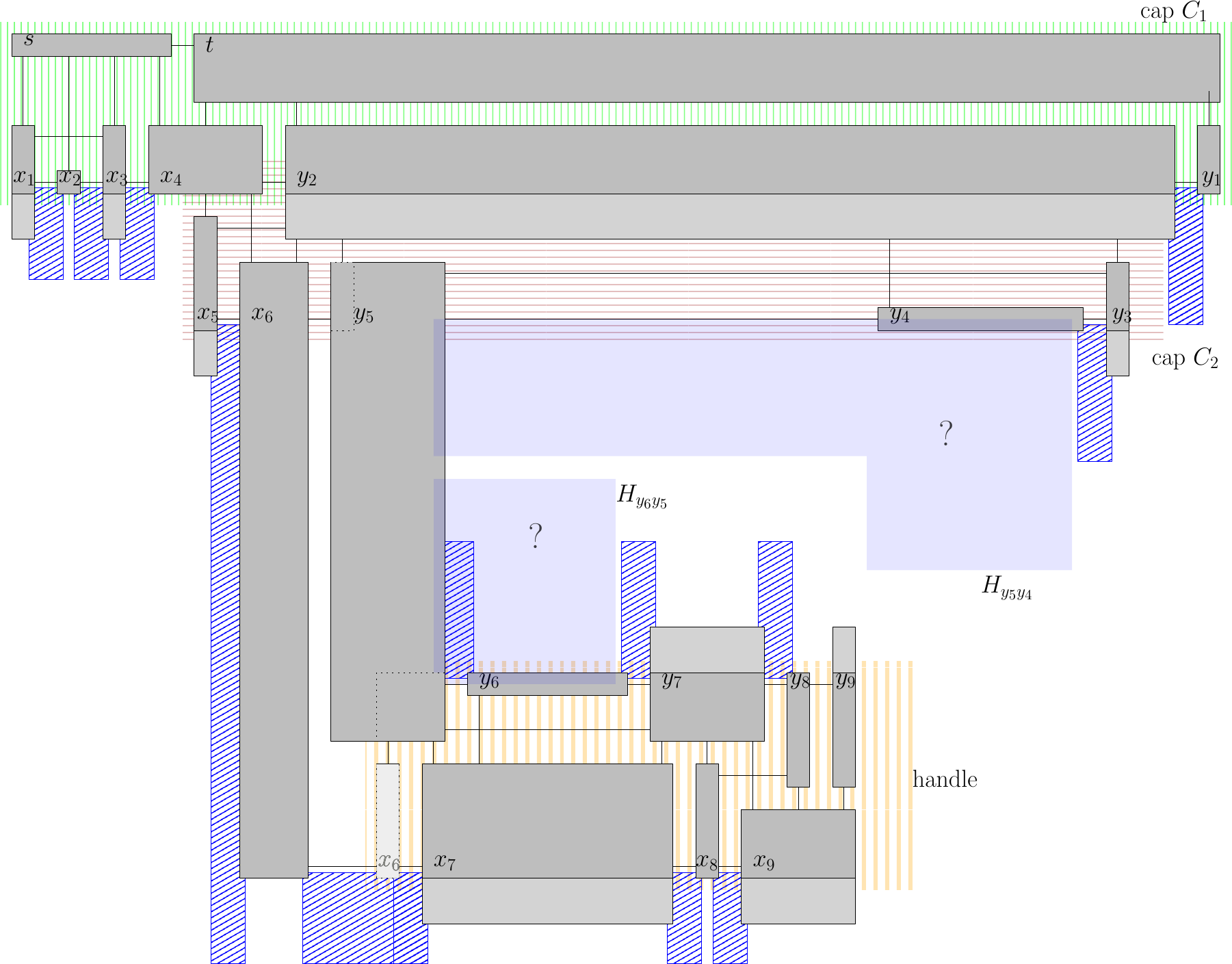}
\caption{Drawing the complete example using its 2-cap umbrella.}
\label{fig:draw2umbrella}
\end{figure}
} 

\onlyabstract{
\begin{figure}[htp!]
{{\includegraphics[width=\textwidth,page=4]{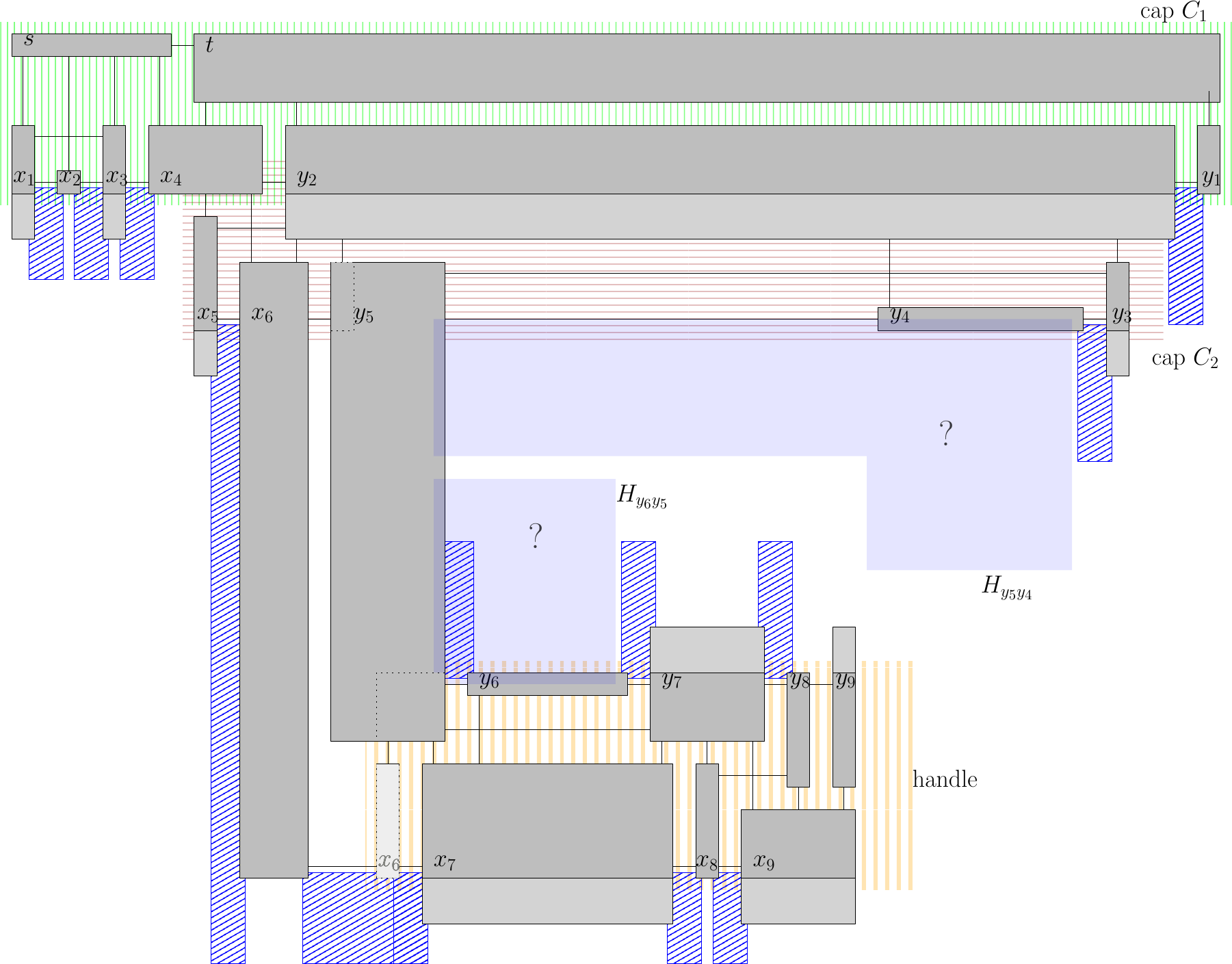}}} 
\caption{Drawing the complete example except for two hanging subgraphs $H_{y_jy_{j+1}}$ and $H_{y_{j-1}y_j}$.  Here $j=5$.} 
\label{fig:draw2umbrellaAbstract}
\label{fig:complete_example}
\end{figure}

\begin{figure}[htp!]
\centering
{\includegraphics[scale=0.48,page=3,trim=250 180 110 180,clip]{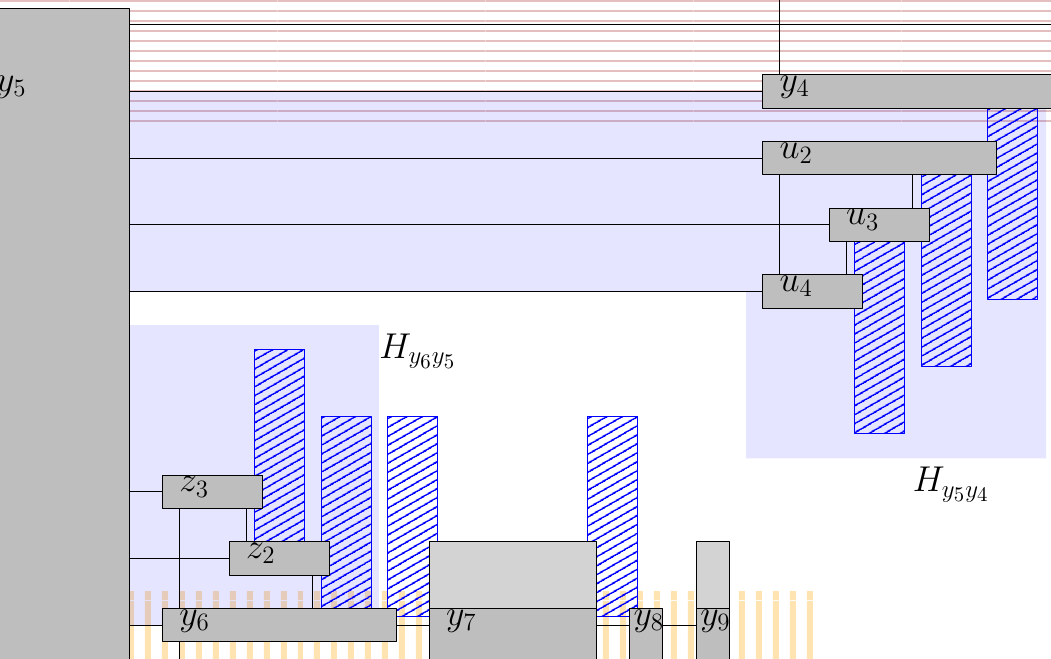}}  
\caption{Closeup on breaking up and merging $H_{y_jy_{j+1}}$ and $H_{y_jy_{j-1}}$.
}
\label{fig:closeup}
\end{figure}
}

\later{
With this we are finally ready for the proof of Lemma~\ref{lem:sqrt}.

\begin{proof}
Set $N=\lceil \sqrt{n} \rceil \leq \sqrt{2F+2}$.
If the last transition edge of the $N$-cap umbrella 
is undefined, then Lemma~\ref{lem:drawkUmbrella} gives a drawing
of height $h(\alpha)+h(\beta)+3N+4\leq h(F)$ as desired.

If the last transition edge is defined, then
consider the caps $C_1,\dots,C_N$ and their transition edges $(X_1,Y_1),\dots,(X_N,Y_N)$. For $k=1,\dots,N$,
let $E_k$ be the set of edges that are incident to $X_k$
or $Y_k$ and do not belong to the $k$-cap umbrella $U^{(k)}$.  Note that for $k<N$, the edges
in $E_k$ either  are in a hanging subgraph of $C_{k}$
or belong to $C_{k+1}$;
either way they do not belong to $E_{k+1}$ because the former edges are not
incident to $X_{k+1}$ or $Y_{k+1}$ and the latter edges belong to $U^{(k+1)}$.
Thus $E_1,\dots,E_N$ are $N\geq \sqrt{n}$ disjoint edge sets.  Since all
of $G$ has at most $3n-6$ edges, therefore $|E_{k^*}|\leq 3\sqrt{n}$ for some 
$1\leq k^*\leq N\leq \sqrt{n}+1$.
\
Apply Lemma~\ref{lem:drawkUmbrella} with $k=k^*$, hence 
and $\max\{d(X_{k^*}),d(Y_{k^*})\}\leq |E_{k^*}|\leq 3\sqrt{n}$.  The resulting embedding-preserving
visibility representation 
has height at most 
\begin{eqnarray*}
h(\alpha)+h(\beta)+3k^*+|E_{k^*}|+4 & \leq &
    h(\alpha)+h(\beta)+3(\sqrt{n}+1) + 3\sqrt{n}+4 \\
    &\leq & h(\alpha)+h(\beta)+6\sqrt{n}+7 \\
    & \leq & h(\alpha)+h(\beta)+11\sqrt{F} + 7 \leq h(F) 
\end{eqnarray*}
since $n\leq 2F+2$ and $6\sqrt{2F+2} \leq 11\sqrt{F}$ for $F\geq 2$.
\end{proof}

So as desired we have constructed embedding-preserving visibility representations of height $h(F)\in O(\sqrt{n})$.  Their width is $O(n)$ can be assumed to be $O(n)$ by Obs.~\ref{obs:widthHeightUpperTrivial}, so the area is $O(n^{1.5})$, and this is asymptotically optimal by Thm.~\ref{thm:lowerGeneral}.  

} 
\onlyfull{We summarize:}

\begin{theorem}
\label{thm:sqrt}
Every $n$-vertex outer-1-plane graph has an embedding-preserving visibility representation of area $O(n^{1.5})$, which is worst-case optimal.
\end{theorem}

\section{Breaking the $\sqrt{n}$-barrier}
\label{sec:breaking}
\abstractlater{\section{Missing details from Section~\ref{sec:breaking}}
}


We know that the height-bound of Theorem~\ref{thm:sqrt} is asymptotically
tight due to Theorem~\ref{thm:lowerGeneral}.  But the lower bound only holds for
embedding-preserving visibility representations---can we get better 
height-bounds if we relax this restriction?  
\onlyfull{The answer turns out to be yes; we sketch a few constructions  here that also use the path of Lemma~\ref{lem:pathTernary}, and Section~\ref{sec:bar} has more constructions that use an entirely different path $\pi$.}

\both{
\begin{theorem}
\label{thm:breaking}
Any outer-1-planar graph $G$ has
\begin{itemize}
\item an embedding-preserving OPVR of complexity 1, and 
\item an embedding-preserving 1-bend orthogonal box-drawing, and 
\item a visibility representation that is not necessarily embedding-preserving and has at most one crossing per edge,
\end{itemize}
and the drawings have area $O(n^{1.48})$.
\end{theorem}
} 
\onlyabstract{We again give the proof mostly in figures.  $(\star)$
We assume as in Section~\ref{sec:sqrt} that the graph is planar-maximal, a reference-edge $(s,t)$ is given, and we construct a $TC_{\sigma,\tau}$-drawing for any given $\{\sigma,\tau\}=\{1,2\}$.
We use $k=1$, i.e., we draw one cap and use the rest of $P_\pi$ as handle.
\todo{TB: added this sentence since we otherwise never said how many caps to use.}
Recall that the main difficulty in Section~\ref{sec:sqrt} was that two hanging subgraphs could not be merged using $TC_{\sigma,\tau}$-drawings since no suitable space was available.  

If we change the drawing model (using a $\Gamma$-shape or a box in the cap-drawing for $y_j$) then one of these hanging subgraphs can use a $TC_{\sigma,\tau}$-drawing, and all edges can still be drawn, perhaps after adding a bend or changing the embedding.  See Fig.~\ref{fig:breaking}.
The other hanging subgraph uses a new drawing-type (i.e., different restrictions on shapes and locations of the endpoints of the reference-edge).
It is not obvious that this exists, but we can show that it can be constructed by adding two rows.  With this, the recursion for the height-function becomes $h(\alpha)+h(\beta)+O(1)$, which resolves to $O(n^{0.48})$ \cite{Chan02}.
}
\begin{figure}[t]
\hspace*{\fill}
\includegraphics[page=10,scale=0.32,trim=100 130 40 60,clip]{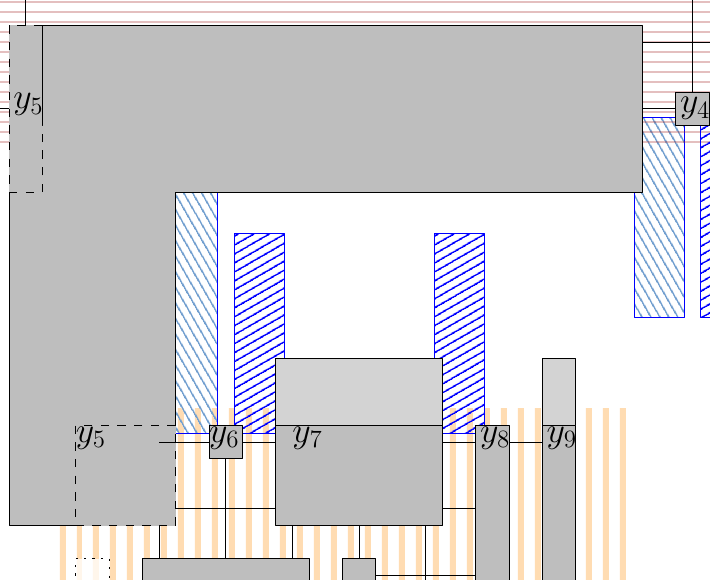}
\hspace*{\fill}
\includegraphics[page=11,scale=0.32,trim=100 130 40 60,clip]{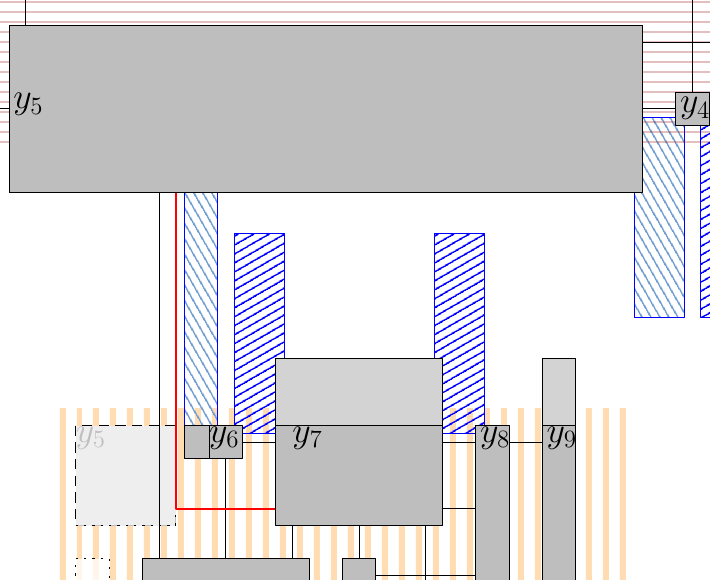}
\hspace*{\fill}
\includegraphics[page=12,scale=0.32,trim=100 130 40 60,clip]{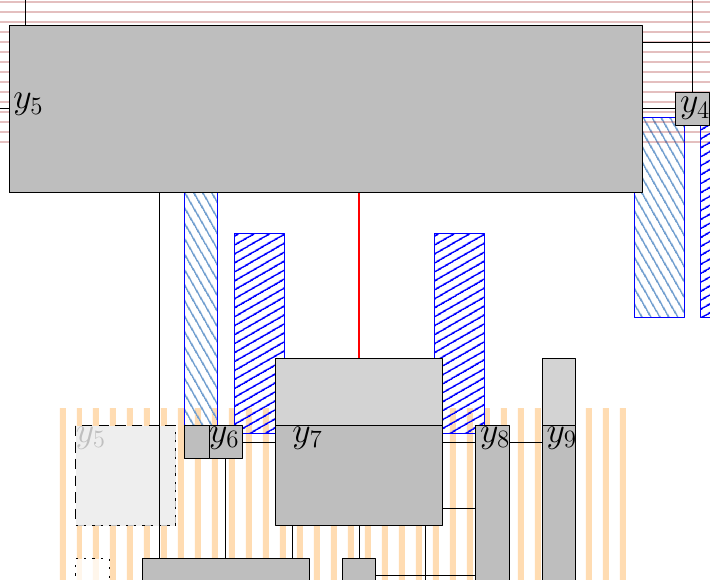}
\caption{Inserting the two remaining hanging subgraphs when permitting ortho-polygons, or bends in edges, or after changing the embedding.
}
\label{fig:breaking}
\end{figure}

\later{
\begin{proof}
By Observation~\ref{obs:widthHeightUpperTrivial} it suffices to achieve height $O(n^{0.48})$.
We give the construction for all three drawing-models simultaneously since they use the same techniques. Let $h'(\cdot)$ be the recursive
function that satisfies $h'(1)=3$ and
$h'(F)=\max_{\alpha^p+\beta^p\leq (1-\delta)F^p} h'(\alpha)+h'(\beta)+9$
(where $\delta>0$ and $p=0.48$ are as in Lemma~\ref{lem:pathTernary}).
We have $h'(n)\in O(n^{0.48})$ \cite{Chan02},  so it suffices
to construct drawings of height $h'(F)$.  

We assume as in Section~\ref{sec:sqrt} that the graph is planar-maximal, a reference-edge $(s,t)$ is given, and we construct a $TC_{\sigma,\tau}$-drawing for any given $\{\sigma,\tau\}=\{1,2\}$.
We use exactly
the same construction as in Section~\ref{sec:sqrt} up until 
Fig.~\ref{fig:drawUmbrella}, so we have an 
embedding-preserving visibility representation for the umbrella and 
merged all hanging subgraphs except $H_{y_{i+1}y_i}$ and $H_{y_iy_{i-1}}$.  We leave $h'(\beta){+}1$ rows free between $(y_j,y_{j-1})$ and the path-drawing $\Gamma_P$; with this the height is $(h'(\alpha){-}1)+(h'(\beta){+}1)+9=h'(F)$.

Let $B_R,B_h$ and $B_U$ be the three boxes that we previously had for $y_j$ in the roof-drawing $\Gamma_R^+$, the handle-drawing $\Gamma_P^+$ and the drawing of the umbrella that had been created.  We define the polygon $P(y_j)$  for $y_j$ and merge hanging subgraphs as follows (we only briefly sketch the approach here and mostly rely on Fig.~\ref{fig:breaking}):  
\begin{itemize}
\item For OPVR-drawings, let $P(y_j)$ by the union of $B_U$ and $B_R$ (widened suitably); this is a $\Gamma$-shaped
	polygon of complexity 1.    With this, $H_{y_jy_{j-1}}$ can
	be merged using a $TC_{2,1}$-drawing.  For $H_{y_{j+1}y_j}$,
	we use what we call a \emph{$\underline{\Gamma}$-drawing}:  $y_{j}$
	is a $\Gamma$-shape that occupies the entire top and left side,
	while $y_{j+1}$ occupies the entire bottom row. 
	(Exactly one of them occupies the bottom left corner, but we will not specify which one as we can easily modify the drawings to achieve either one.)    See Fig.~\ref{fig:drawingTypes}.
	It is not trivial that $\underline{\Gamma}$-drawings exist; we will explain how to create one of height $h'(\beta)+2$ below.  The top and bottom row of the drawing of $H_{y_{j+1}y_j}$ can re-use rows of $\Gamma_R^+$ and $\Gamma_P$, so this fits (after stretching vertically, if needed) into the rows available. 
\item For 1-bend orthogonal box-drawings, 
	$P(y_{j})$ is $B_R$ (widened), and we re-route the edges incident to $B_h$. 
	No edges attach at the left or top side of $B_h$ since we removed $H_{y_jy_{j+1}}$.  Any edges at the bottom side of $B_h$ can simply be extended vertically upward to reach $P(y_j)$.  As for edges on the right side of $B_h$, observe first that $B_h$ has height 2 since $x_i$ was drawn with height 3.
    So there can be at most two edges on the right side of $B_h$; one is $(y_j,y_{j+1})$ while the other (if any) is $(y_j,y_{j+2})$.
	Edge $(y_j,y_{j+1})$ can be redrawn vertically.  If $(y_j,y_{j+2})$ exists,
	then we draw it with a bend placed in $B_h$.  We can merge $H_{y_{j},y_{j-1}}$ easily, but for $H_{y_{j+1}y_j}$ we use what
	we call a \emph{$BB$-drawing}:  $y_j$ and $y_{j+1}$ are bars that occupy the entire top
	row and the entire bottom row, respectively.  Again we will need
	to argue below that this actually exists with height $h'(\beta)+2$.
\item For visibility representations, we proceed as above,
	but route $(y_j,y_{j+2})$ (if it exists) by going vertically
	upward from $y_{j+2}$.  (Note that this changes the embedding.)
\end{itemize}

Consider the face $f_{st}$ of $\skel{G}$, which as discussed in the proof of Lemma~\ref{lem:drawPath} can have one of five configurations.
We explain one easy case in detail and rely on Fig.~\ref{fig:cases} for the others.  
Let us assume that $f_{st}$ does not contain a crossing 
(so it is a triangle $\{s,t,x\}$) and that $(t,x)$ belongs to 
the next face $f'$ of path $\skel{P_\pi}$.  Let $G'$ be the hanging subgraph
at $(x,t)$, recursively obtain a $TC_{2,1}$-drawing $\Gamma'$ of $G'$, and modify it into a $TB_{2,1}$-drawing.    
Insert a new row below $\Gamma'$ into which we insert $P(s)$, and expand $P(x)$ leftward.
This gives an embedding-preserving $BB$-drawing after connecting edges vertically and
inserting a (recursively obtained) $BB$-drawing of the hanging subgraph $H_{sx}$.
If instead we want an $\underline{\Gamma}$-drawing, then we expand $t$
and also $x$ into a $\Gamma$-shape and merge an $\underline{\Gamma}$-drawing
of $H_{sx}$.
See Fig.~\ref{fig:cases}(a-c) for this and all other cases.

\def\scalefactor{0.52}
\begin{figure}[ht!]
\centering
\begin{tabular}{ccccc}
& $\underline{\Gamma}$-drawing, & $BB$-drawing, & $BB$-drawing, & 
\\
Case& ortho-polygon & permit a bend& change embedding & 
\\
\includegraphics[page=10,scale=\scalefactor]{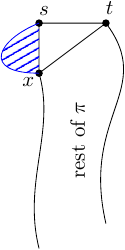}
&  ~~~\includegraphics[page=11,scale=\scalefactor]{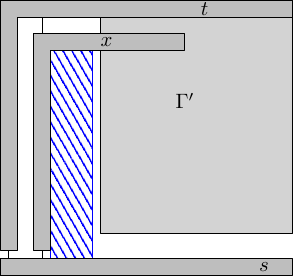}~~~
&  \includegraphics[page=12,scale=\scalefactor]{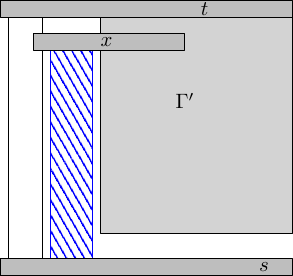}~~~
&  \includegraphics[page=12,scale=\scalefactor]{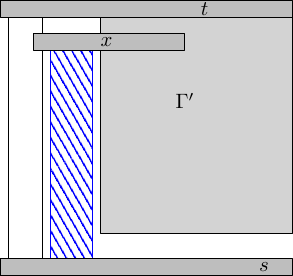}
\\[1.5ex]
\includegraphics[page=13,scale=\scalefactor]{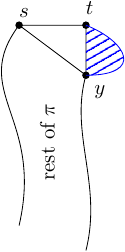} 
& \includegraphics[page=14,scale=\scalefactor]{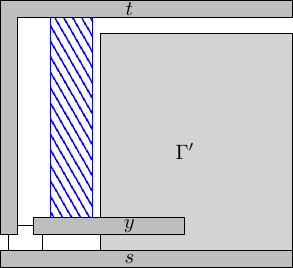}
& \includegraphics[page=15,scale=\scalefactor]{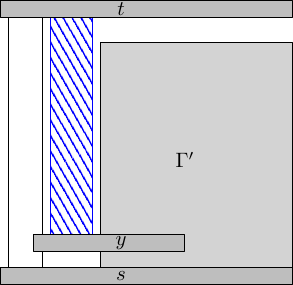}
& \includegraphics[page=15,scale=\scalefactor]{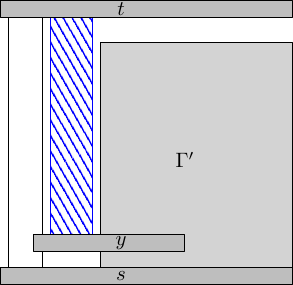}
\\[1.5ex]
\includegraphics[page=16,scale=\scalefactor]{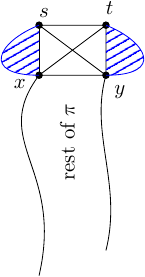} 
& \includegraphics[page=17,scale=\scalefactor]{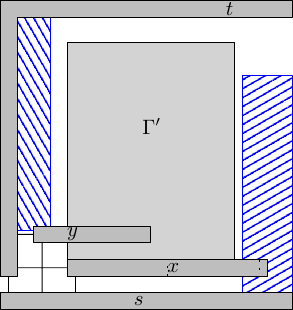}
& \includegraphics[page=18,scale=\scalefactor]{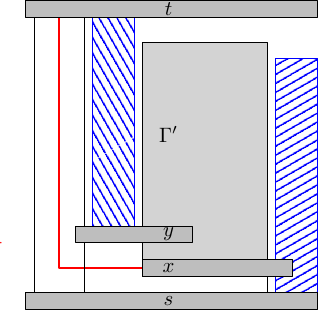}
& \includegraphics[page=19,scale=\scalefactor]{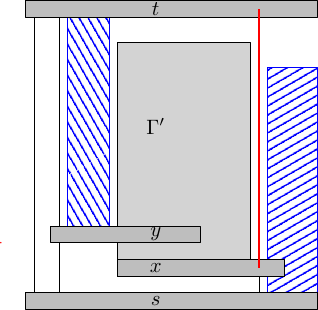}
\\[1.5ex]
\includegraphics[page=20,scale=\scalefactor]{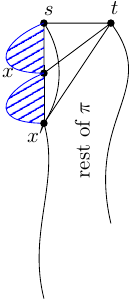} 
&\includegraphics[page=21,scale=\scalefactor]{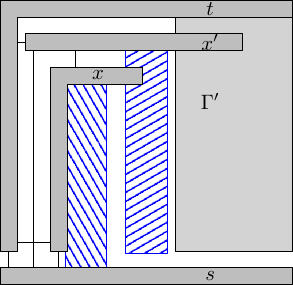} 
&\includegraphics[page=22,scale=\scalefactor]{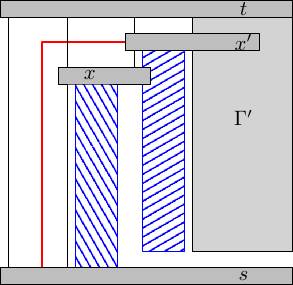} 
&\includegraphics[page=23,scale=\scalefactor]{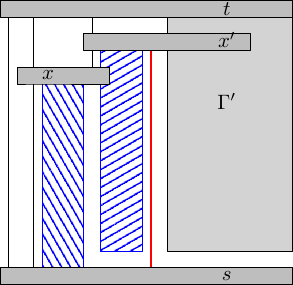} 
\\[1.5ex]
\includegraphics[page=24,scale=\scalefactor]{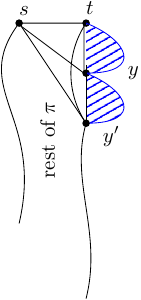} 
&\includegraphics[page=25,scale=\scalefactor]{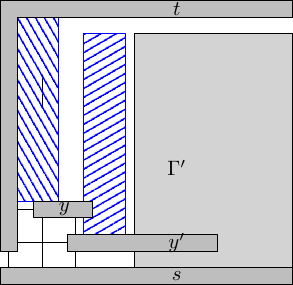} 
&\includegraphics[page=26,scale=\scalefactor]{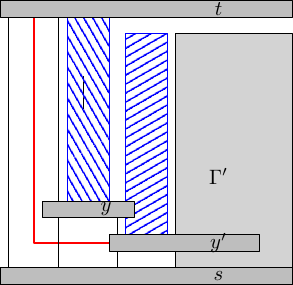} 
&\includegraphics[page=27,scale=\scalefactor]{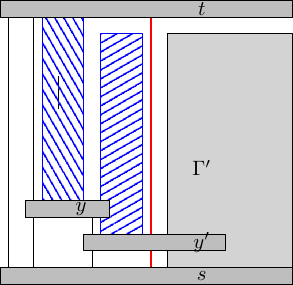} 
\\
& (a) & (b) & (c) & 
\end{tabular}
\caption{Creating (a) $\underline{\Gamma}$-drawings, (b) and (c) $BB$-drawings.
Upward-striped hanging subgraphs use a $TB_{\sigma,\tau}$-drawing, while downward-striped hanging subgraphs use an $\underline{\Gamma}$-drawing or a $BB$-drawing, respectively.
}
\label{fig:cases}
\end{figure}

The height of $\Gamma'$ is at most $h'(F)$.
We always add at most two rows above/below $\Gamma'$, so the height requirement in the columns containing $\Gamma'$ is at most $h'(F)+2$.  We may also use up to $\max\{h'(\alpha),h'(\beta)\}+4$ rows at the hanging subgraph(s), but this is at most $h'(F)$.  So the drawing has height at most $h'(F)+2$ as required.
%
\end{proof}
} 


\section{Optimum-height drawings in other drawing models}
\label{sec:bar}

\abstractlater{
\section{Missing details from Section~\ref{sec:bar}}
\label{app:bar}
}

In this section we give drawings whose height (and area) is also optimal, but they are in a different drawing model (hence different lower bounds apply).

\subsection{Embedding-preserving bar visibility representations}
\abstractlater{
\subsection{Embedding-preserving bar visibility representations}
}
We proved in Theorem~\ref{thm:lowerBVR} that any embedding-preserving bar-visibility-represen\-ta\-tion has height $\Omega(n)$ for some outer-1-plane graphs.  
A fairly straight-forward greedy-construction shows that we can match this.   The main difficulty is showing that such a drawing exists as all; the area-bound then follows from Obs.~\ref{obs:widthHeightUpperTrivial}). $(\star)$

\both{
\begin{theorem}
\label{thm:barVR}
Any outer-1-planar graph $G$ has an embedding-preserving  bar-visibility representation of area $O(n^2)$ which is worst-case optimal. 
\end{theorem}
}
\later{
\begin{proof}
By Obs.~\ref{obs:widthHeightUpperTrivial} the area-bound holds trivially, as long as we show that there always \emph{exists} an embedding-preserving bar visibility representation for outer-1-planar graphs.  (This is not trivial: some 1-planar graphs have no embedding-preserving visibility representations at all \cite{BFL-DCG} and in particular no bar visibility representations.)

Fix an arbitrary reference-edge $(s,t)$ and an arbitrary root-to-leaf path $\pi$.  
We show that $G$ in fact has two such drawings, one is a $TB_{1,2}$-drawing while the other is a $TB_{2,1}$-drawing.  
Consider the five possible configurations at the face $f_{st}$ incident to $(s,t)$ in $\skel{G}$ that we saw in Fig.~\ref{fig:cases}.  Let $G'$ be the hanging subgraph at the edge that $f_{st}$ has in common with the next face of $\skel{P_\pi}$.  Recursively obtain either a $TB_{1,2}$-drawing or a $TB_{2,1}$-drawing $\Gamma'$ of $G'$ (this choice is dictated by the configuration of $f_{st}$).   In all cases, we can add the missing vertices of $f_{st}$ and insert (recursively obtained) drawings of hanging subgraphs at edges of $f_{st}$ to obtain the desired bar visibility representation.
Fig.~\ref{fig:barVR} shows all cases for creating a $TB_{2,1}$-drawing; creating a $TB_{1,2}$-drawing is symmetric.
\end{proof}
}

\later{
\def\scalefactor{0.55}
\begin{figure}[ht]
\includegraphics[page=1,scale=\scalefactor]{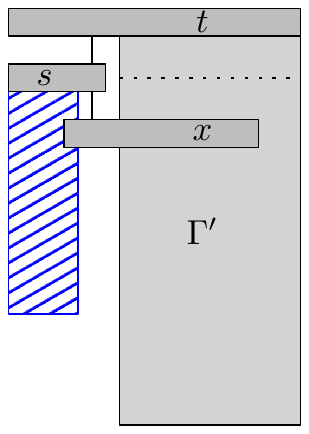}
\hspace*{\fill}
\includegraphics[page=2,scale=\scalefactor]{barVR.pdf}
\hspace*{\fill}
\includegraphics[page=3,scale=\scalefactor]{barVR.pdf}
\hspace*{\fill}
\includegraphics[page=4,scale=\scalefactor]{barVR.pdf}
\hspace*{\fill}
\includegraphics[page=5,scale=\scalefactor]{barVR.pdf}
\caption{Constructing an embedding-preserving visibility representation.}
\label{fig:barVR}
\end{figure}
}

\subsection{More lower bounds} 

We now prove other lower bounds on the height that depend
on the \emph{pathwidth} $pw(G)$ and the number of crossings $\chi(G)$ of the outer-1-plane graph.

\onlyfull{\subsubsection{Lower bound based on pathwidth. }}
We recall that a {\em path decomposition} of a graph $G$
consists of a collection $B_1,\dots,B_\xi$ of vertex-sets
(``bags'') such that every vertex belongs to a consecutive set of bags, and
for every edge $(u,v)$ at least one bag contains both $u$ and $v$.
The {\em width} of such a path decomposition is $\max_i \{|B_i|-1\}$,
and the {\em pathwidth} $pw(G)$ is the minimum width of a path
decomposition of $G$.    Any outer-1-planar graph has pathwidth $O(\log n)$, since it has treewidth 3 \cite{Auer16}.

For planar drawings, the width and height of a drawing is lower-bounded by the pathwidth of the graph \cite{FLW03}. 
%

Less is known for non-planar drawings.  It follows from the proof
of Corollary~3 in \cite{BCK+-MFCS} that any bar-1-visibility representation of graph $G$ has height at least $pw(G)+1$.  
\onlyabstract{Roughly speaking,  we can extract a path decomposition of $G$ by scanning left-to-right with a vertical line and attaching a new bag whenever the set of intersected vertices changes.}
We use the same proof-idea here to show a lower bound for \emph{all} OP-$\infty$-drawings.  $(\star)$

\abstractlater{\subsection{Lower bounds}}
\both{
\begin{theorem}
\label{thm:pwLowerBounds}
Any OP-$\infty$-drawing of a graph $G$ (not necessarily outer-1-planar) has height and width $\Omega(pw(G))$.
\end{theorem}
}%
\later{%
\begin{proof}
We only prove the bound on the height $h$; the bound on the width then holds after rotating $\Gamma$ by $90^\circ$.
Define a new graph $G'$ as follows:  For any edge $e=(u,v)$ of $G$ that is drawn with $k$ horizontal segments, replace $e$ by a path from $u$ to $v$ with $k$ new vertices.  For each such new vertex $x$, let $P(x)$ be the horizontal segment that $x$ corresponds to.  Since vertex-polygons are disjoint and no edge-segment overlap, any grid point now belongs to (at most) three polygons $P(\cdot)$; one from a vertex of $G$ and up to two from horizontal edge-segments of $\Gamma$.  

We obtain a path decomposition $\mathcal{P}$ of $G'$ by sweeping a vertical 
line $\ell$ from left to right.  We interrupt the sweep whenever $\ell$ 
reaches the $x$-coordinate of a vertical edge-segment or a vertical
side of a vertex-polygon.  At this time-point, attach a new bag $B$
at the right end of $\mathcal{P}$ and insert all vertices $v$ for which $P(v)$ is intersected by $\ell$.  The properties of a path decomposition are easily 
verified: For any vertex $v$ of $G'$, polygon $P(v)$ span a contiguous set of 
$x$-coordinates, and hence $v$ belongs to a contiguous set of bags.  
Any edge $(v,w)$ of $G'$ is represented by a vertical segment, and hence
covered by the bag created when sweeping the $x$-coordinate of this segment.
Finally any bag has size at most $3h$ since at most three polygons occupy each grid-point in a column.  Since $G'$ is a subdivision of $G$, therefore $pw(G)\leq pw(G')\leq 3h-1$ as required.
\end{proof}

With more care (and after inserting columns and thicken vertex-polygons to have no vertical segments) one can actually show $pw(G')\leq 2h-1$, and $pw(G')\leq h$ if no horizontal edge-segment intersects a vertex-polygon without ending there. We leave those details to the reader.
}
By the remark after Obs.~\ref{obs:widthHeightLowerTrivial} hence some outer-1-planar graphs require area $\Omega(n\, pw(G))$  in all OP-$\infty$-orthogonal drawings.

\onlyfull{\subsubsection{Lower bound based on crossings}}

If we specifically look at drawings that have no crossings, then we can also create a lower bound based on the number of crossings.
\onlyabstract{This is easily obtained by modifying the lower-bound example from \cite{Bie-GD20}. $(\star)$}
\both{
\begin{theorem}
\label{thm:crossingLowerBounds}
For any $k$ and $n\geq 4k$, there exists an outer-1-plane graph with 
$n$ vertices and $k$ crossings that requires at least $2k$ height and width in any planar drawing.
\end{theorem}
}
\later{
\begin{proof}
This essentially follows from the proof of \cite[Theorem 1]{Bie-GD20}. Consider the graph with $4k$ vertices $\{x_1, x_2, \ldots x_{2k}, y_{2k}, y_{2k-1},\ldots, y_1\}$ as the outer-face
where $\{x_{2i-1},x_{2i},y_{2i-1},y_{2k}\}$ form a complete graph $K_4$ 
drawn with one crossing for $i=1,\dots,k$.  
This is an outer-1-path
with crossing edges in every second inner face of the planar skeleton. As argued in \cite{Bie-GD20} any planar drawing of this graph contains $k$
nested triangles, and thus must have height at least $2k$.
If $4k<n$, then adding an arbitrary $n{-}4k$ vertices gives the desired graph.
\end{proof}

\begin{figure}[ht]
\hspace*{\fill}
\includegraphics[page=3,width=0.7\linewidth]{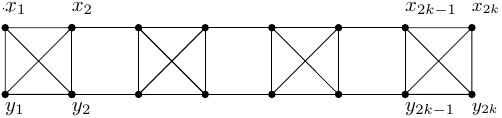}
\hspace*{\fill}
\caption{A graph with $k$ crossings that requires height at least $2k$ in any planar drawing.  }
\end{figure}

Note that the graph used in the example above has pathwidth $3$; thus this crossing number lower bound is not subsumed by the pathwidth lower bound.
}
In particular, the lower bound on the height of a planar OP-$\infty$-orthogonal drawing of $G$  is $\Omega(\max\{\allowbreak pw(G),\allowbreak \chi(G)\})$, which is the same as $\Omega(pw(G)+\chi(G))$.

\todo{For long version or revised version, re-insert complexity-4 result here.}

\subsection{More constructions}

We now turn towards creating bar representations that prove that Theorem~\ref{thm:pwLowerBounds} and \ref{thm:crossingLowerBounds} are tight.  

\abstractlater{\subsection{More constructions}}
\abstractlater{
We now give the details that lead to Theorem~\ref{thm:bar1VR}.
} 

\later{
As in Section~\ref{sec:sqrt} we assume that our input graph $G$ is plane-maximal with a reference-edge $(s,t)$, and we choose some root-to-leaf path $\pi$ in $\dual{G}$.
Let $P_\pi$ be the outer-1-path whose inner dual is $\pi$ and let $f_{st}$ be the face of 
$\skel{P_\pi}$ incident to $(s,t)$.   At most one of $s$ and $t$ can belong to the face $f'$ after $f_{st}$ on $\skel{P_\pi}$; up to symmetry we assume that $s$ does not belong to $f'$.  
The \emph{enhanced $\pi$-path} $P_\pi^+$ is 
the plane-maximal outer-1-path 
formed by $P_\pi$ together with all neighbours of $s$.
Let $\hat{e}\neq (s,t)$ be an arbitrary outer-face edge on the leaf-face of $\pi$, and enumerate the outer face of $P_\pi^+$ as $x_0,\dots,x_\ell,y_r,\dots,y_0$ where $\hat{e}=(x_\ell,x_r)$ and $(s,t)=(y_0,y_1)$.  Note that (in contrast to Section~\ref{sec:sqrt}) the end-edge is \emph{not} $(s,t)$; instead it is the other outer-face edge incident to $s$.

\begin{figure}[ht]
\hspace*{\fill}
\includegraphics[scale=0.4,page=2]{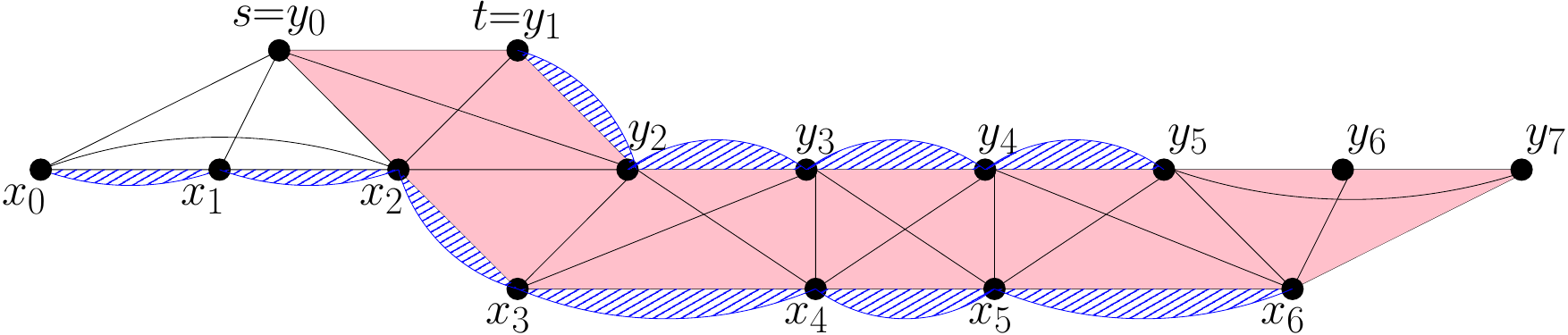}
\hspace*{\fill}
\caption{Example-graph for constructing bar visibility representations.  $P_\pi$ is shaded.
}
\end{figure}

\subsubsection{Drawing $P_\pi^+$.}

}

\begin{figure}[htp!]
\centering
\subfigure[\label{fig:bar_ex}]
{\includegraphics[scale=0.35,page=1]{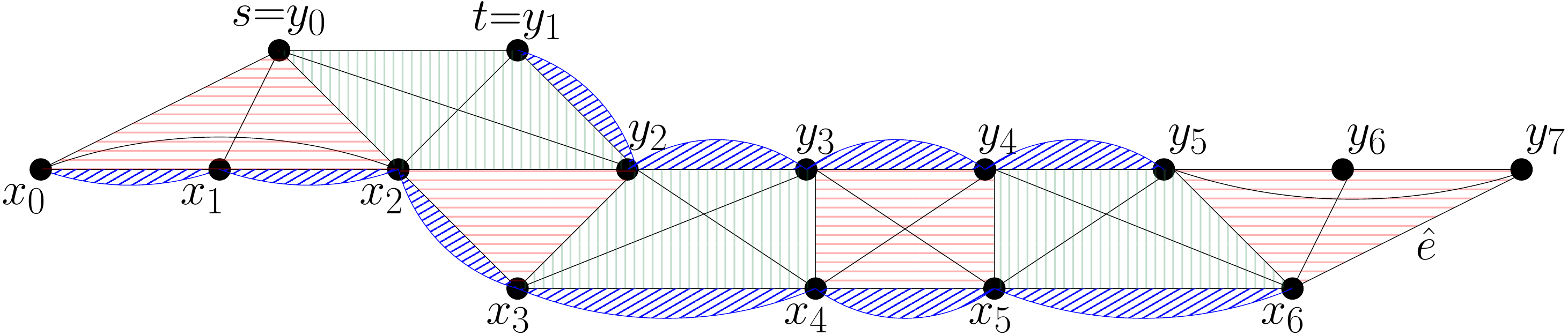}}\\
\subfigure[\label{fig:outerpathAbstract}]{\includegraphics[scale=0.4,page=5]{bar_example.pdf}}\\
\subfigure[\label{fig:draw_bar_ex}] {\includegraphics[scale=0.4,page=3]{bar_example.pdf}}\\
\subfigure[\label{fig:draw_planar_ex}] {\includegraphics[scale=0.4,page=8]{bar_example.pdf}}\\
\caption{(a) Example-graph.  (b) Drawing its skeleton and merging hanging subgraphs inward.  Some vertex-boxes are artificially wide to match (c).
(c) The bar-1-visibility representation obtained by moving some bars and sometimes traversing bars.
(d) The planar bar-visibility representation obtained by moving some bars and extending them rightwards.
Arrows indicate vertices that get moved outward beyond their neighbour on the right.
}
\end{figure}

\later{
We already discussed in Lemma~\ref{lem:drawPath} how to create a visibility representation of an outer-1-path, but now the situation is different because we want bars, rather than boxes.  
The \emph{canonical drawing} of $\skel{P_\pi^+}$ is defined as follows (we adopt ideas from \cite{Bie-WAOA12}, but use an extra row). 
All bars for $x_0,\dots,x_\ell$ are in row $-1$; all bars for $y_1,\dots,y_r$ are in row $1$.  We assign $x$-coordinates to the endpoints of these bars by parsing 
the faces $f_1,\dots,f_k$ of $\skel{P_\pi^+}$ in order, beginning with the one incident to $(x_0,y_0)$.  
We use $b(u)$ for the bar representing $u$.
Both $b(x_0)$ and $b(y_0)$ begin at $x$-coordinate 1.  If face $f_h$ (for $h\geq 1$) contains no crossing, say $f_h=\{x_i,x_{i+1},y_j\}$, then $b(x_i)$ ends, $b(x_{i+1})$ begins and $b(y_j)$ extends further rightwards.  See e.g.~$\{x_2,x_3,y_2\}$ in Fig.~\ref{fig:outerpath}.  If $f_h$ contains an opposite-boundary crossing, say $f_h=\{x_i,x_{i+1},y_j,y_{j+1}\}$, then 
$b(x_i)$ ends and $b(x_{i+1})$ begins, and one unit further right $b(y_j)$ ends end $b(y_{j+1})$ begins, see e.g.~$\{x_4,x_5,y_3,y_4\}$ in Fig.~\ref{fig:outerpath}.
Finally if $f_h$ contains a same-boundary crossing, say $f_h=\{x_i,y_j,y_{j+1},y_{j+2}\}$, then $b(y_j)$ ends, $b(y_{j+1})$ begins and ends immediately again, and then $b(y_{j+2})$ begins,   see e.g.~$\{x_6,y_5,y_6,y_7\}$ in Fig.~\ref{fig:outerpath}.

\begin{figure}[ht]
\hspace*{\fill}
\includegraphics[scale=0.4,page=6]{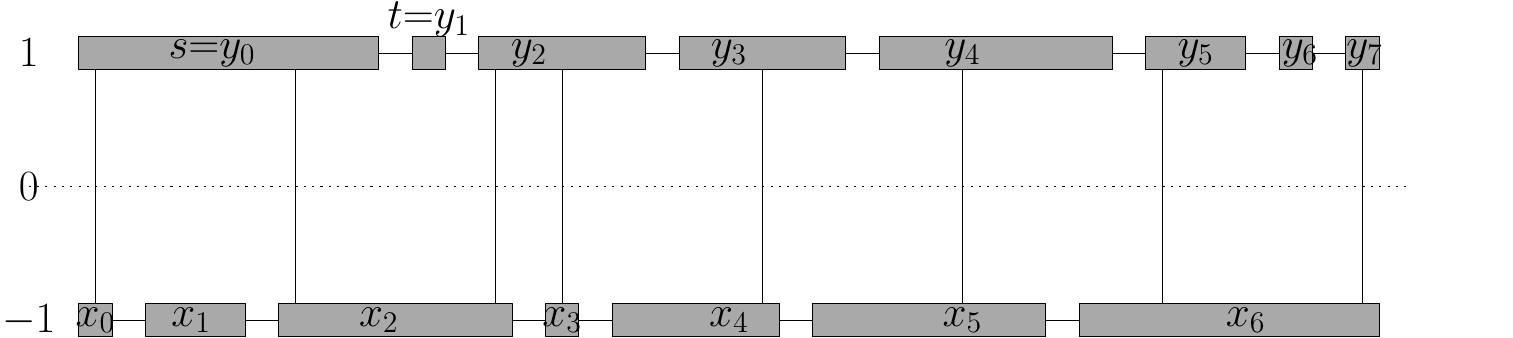}
\hspace*{\fill}
\caption{Drawing the planar skeleton.}
\label{fig:outerpath}
\end{figure}


Observe that row 0, which we call the {\em center row}, remains {\em empty}, i.e.,  intersects no vertex.  (We can delete empty rows in the final drawing, but maintaining it during the construction gives us a place to add further rows as needed.)

The resulting drawing $\Gamma_P$ represents $\skel{P_\pi^+}$, and, in order to draw $P_\pi^+$, we shall draw each pair of crossing edges. The method to do this depends on whether we have same-boundary crossings or opposite-boundary crossings and on the drawing model, and 
will be detailed in the following sections.

\subsubsection{Inward merging.}
As in Section~\ref{sec:sqrt}, the idea is to merge hanging subgraphs of the path, but this time we merge ``inward''.  Assume that for each hanging subgraph we can recursively obtain a $TC_{1,1}$-drawing $\Gamma_H$, i.e., 
bars $b(u)$ and $b(v)$ occupy the top corners.
Expand the center row into as many rows as needed for the hanging subgraphs.  Then for each	hanging subgraph $H$, we insert $\Gamma_H$ (possibly after rotation) 	in these new rows in such a way that the two drawings of the reference-edge of $H$ coincide.  
See Fig.~\ref{fig:inward_merging}.
We remark that, when performing an inward merging, the outer-1-planar embedding is not respected (some vertices will not be on outer-face anymore).
	Inward merging was used in \cite{Bie-WAOA12} to create 
	drawings of outer-planar graphs with height $O(pw(G))$.

\begin{figure}[ht]
\hspace*{\fill}
\includegraphics[scale=0.4,page=5]{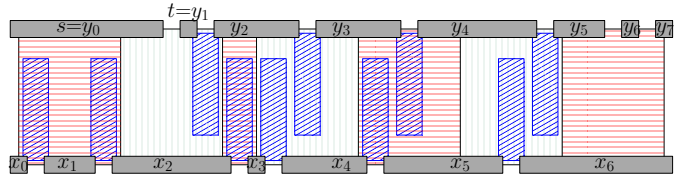}
\hspace*{\fill}
\caption{Inward merging.  We shade the faces of $\skel{P_\pi^+}$ for easier identification later on.}
\label{fig:inward_merging}
\end{figure}



\subsubsection{Handling most same-boundary crossings.}

Let $\Gamma_P$ be the canonical drawing of $\skel{P_\pi^+}$; here vertices $s,t$ are drawn at the left end of row 1 since $(s,t)=(y_0,y_1)$.
We now show how to handle all same-boundary crossings except the one that may occur at
 face $f_{st}$.  Add a new row each above
and below the center-row.  For a same-boundary crossing, let the \emph{middle vertex} be the middle of the three vertices that are on the same boundary.   For each same-boundary crossing (say at face $f\neq f_{st}$), move
the middle vertex into the adjacent new row and re-route the incident edges of 
$f$ vertically  after extending vertex bars as needed.  The two crossing
edges can now both be inserted; one horizontally in row 1 or $-1$, and the
other one vertically.    See Fig.~\ref{fig:sameSideCrossing}.

\begin{figure}[ht] \centering
\centering
\includegraphics[scale=0.4,page=7]{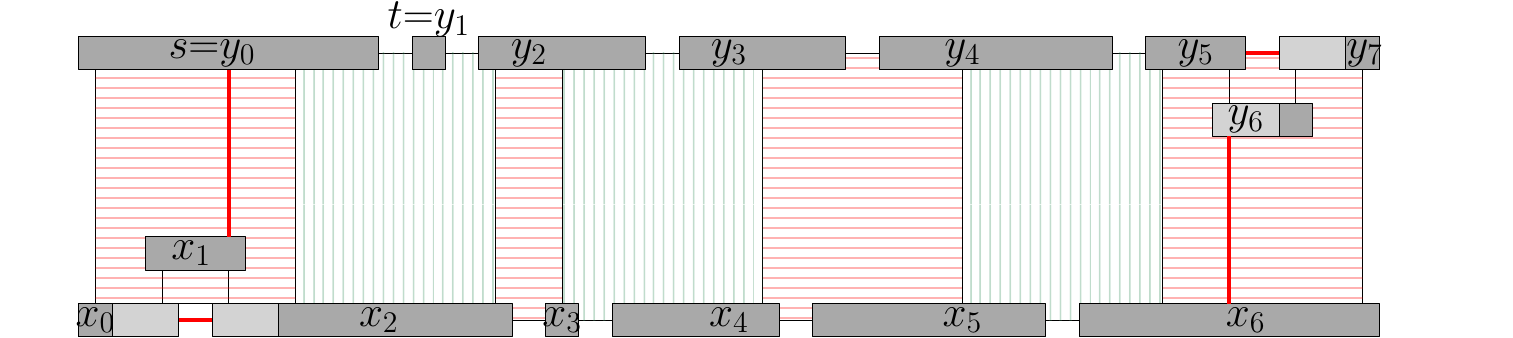}
\caption{Handling most same-boundary crossings.  Edges in crossings are
red (bold).  Light gray boxes are extensions of existing vertex-boxes.}
\label{fig:sameSideCrossing}
\end{figure}

We must be careful to keep space available for merging hanging subgraphs later.  We say that an edge $(u,v)$ \emph{permits merging} in one of the following two cases.  Either $(u,v)$ is drawn horizontally, and the axis-aligned rectangle between $(u,v)$ and the center-row is empty.  Alternatively, $(u,v)$ is drawn vertically, both its ends are on the same side of the center-row, and if (say) $v$ is closer to the center-row than $u$, then $b(v)$ ends at $(u,v)$, and there is an axis-aligned empty rectangle where one horizontal side is $u$, the other horizontal side is on the center, and one vertical side contains $(u,v)$.
In the canonical drawing all edges that might have hanging subgraphs were drawn horizontally and permitted merging.
One verifies that edges that are now drawn vertically (due to a same-boundary crossing) likewise permit merging.

\subsubsection{Handling the remaining crossings.}

Now consider an opposite-boundary crossing, say at face $\{x_i,$ $x_{i+1},$ $y_j,$ $y_{j+1}\}$.
None of $\{x_i,$ $x_{i+1},$ $y_j,$ $y_{j+1}\}$ is the middle vertex
of a same-boundary crossing, so these vertices are in rows 2 and $-2$, respectively.
Handling opposite-boundary crossings, as well as the same-boundary crossing
at $f_{st}$, is the only part where the drawing-algorithm depends on the drawing-model.

\paragraph{Planar bar visibility representations.}

For planar bar visibility representations, we handle opposite-boundary crossings
by moving vertices outward to new rows, and inserting one of each pair
of crossings at the extreme right end of the drawing.  
Specifically, 
at any opposite-boundary crossing at a face $\{x_h,x_{h+1},y_z,y_{z+1}\}$,
	mark $x_h$ and $y_{z+1}$.
Also, if $f_{st}$ is a same-boundary crossing, then it is of the form  $f_{st}=\{s,t,y_1,x_\rho\}$ for some index $\rho$ and we  mark $y_1$ and $x_\rho$.
Finally, mark $t$ if it was not marked yet.

See Fig.~\ref{fig:oppositeBoundaryCrossingBig} for the following process.
For any vertex $v$, let $m(v)$ be the total number of marks at $v$ or at vertices further right on the same boundary as $v$.
We move $y_z$ upwards by $m(y_z)$ rows and we move $x_h$ downwards by $m(x_h)$ rows (inserting new rows as needed and re-routing outer-face edges so that they permit merging).
If $\chi_P$ is the number of crossings in $P_\pi^+$,  then the total number of rows is now at most $4+2\chi_P$:  We started
with 3 rows originally, added one row because we marked $t$, add two rows total for same-boundary crossings (but need those only if there are actually such crossings), and then two more rows for each remaining crossings.

We restore edges and insert the missing crossings 
in order from right to left, i.e., by decreasing index on the boundary.
Consider an opposite-boundary-crossing at face $f=\{x_h,x_{h+1},y_z,y_{z+1}\}$.
Vertices $x_h$ and  $y_{z+1}$ was marked, hence $b(y_{z+1})$ is higher
than the bars of $y_{z+2},\dots,y_r$ while $b(x_h)$ is lower than the bars of $x_{h+1},\dots,x_\ell$.  We can hence expand both $y_{z+1}$ and $x_h$ rightward beyond everything drawn in the rows between them and insert $(x_{h},y_{z+1})$ as a vertical edge at the very right end of the drawing.  We can also draw $(x_{h+1},y_{z})$ vertically within $f$, after expanding $x_h$ and $y_{z+1}$ towards each other if needed.

\begin{figure}[ht] \centering
\hspace*{\fill}
\includegraphics[scale=0.4,page=8]{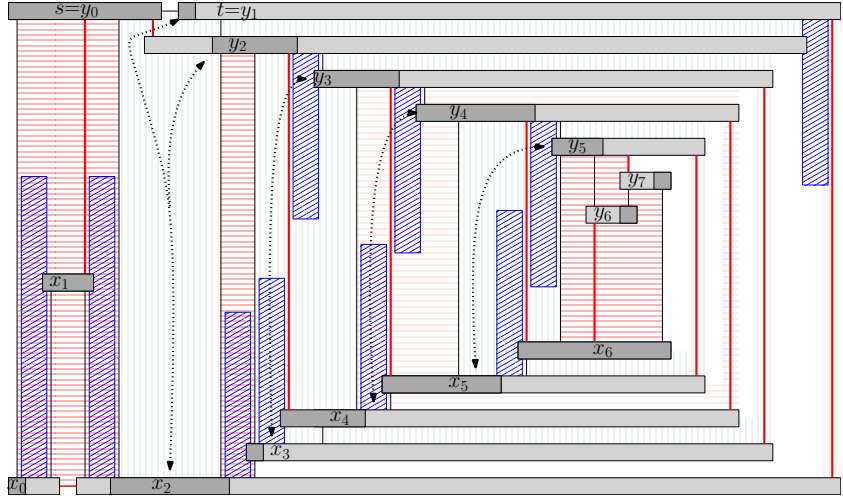}
\hspace*{\fill}
\caption{Handling the remaining crossings by adding rows for a planar bar visibility representation. Dotted arrows indicate vertices that are marked.
}
\label{fig:oppositeBoundaryCrossingBig}
\end{figure}

If $f_{st}$ is a same-boundary crossing, then
$t$ and $x_{\rho}$ were both marked,
and so can be expanded rightward and $(t,x_\rho)$ can be drawn vertically at the very right end of the drawing. 
Also, edge $(s,y_2)$ is drawn  vertically within $f_{st}$, by expanding $s$ and $y_2$ towards each other. 
We also marked $y_2$, and did this so that $b(y_2)$ likewise can be expanded
to the right; then outer-face edge $(y_2,t)$  can be drawn at the right end to permit merging.  

Note that $t$ is always marked and that $s$ is never marked, so $m(t)=m(s)$ and $s,t$ 
are both
in the topmost row and no other vertex is in that row.
If needed, we can expand $t$ to cover the rightmost endpoint.  We hence have:

\begin{lemma}
\label{lem:pathFewCrossings}
For any root-to-leaf path $\pi$ in $\dual{G}$,
graph $P_\pi^+$
has a planar bar visibility representations on at most $4+2\chi(P_\pi^+)$ rows that is a $C_{1,1}$-drawing.

Furthermore, the center row is empty, and
any edge that could have an attached hanging subgraph permits merging.
%
\end{lemma}

\paragraph{Bar-1-visibility representations.}

For bar-1-visibility representations, 
we change the drawing as illustrated in Fig.~\ref{fig:oppositeBoundaryCrossingBar1VR}.
First we add three new rows above and
one new row below, so that we now have rows $-3,\dots,5$.
Parse the bars in row $-2$ left to right (these are the bars $b(x_0),\dots,b(x_\ell)$, except those of middle vertices of opposite-boundary crossings).  Move every second of these bars into row $-3$.  We can choose to move $b(x_0)$ or not, and do this choice as follows.  If $f_{st}$ contains an opposite-boundary crossing, say $f_{st}=\{x_\rho,x_{\rho+1},s,t\}$ then we do the choice such that $x_\rho$ ends in row $-3$; otherwise the choice is arbitrary.  Any outer-face edge $(x_h,x_{h+1})$ can be re-routed vertically (after extending bars of its endpoints towards each other) and then still permits merging.

As for the bars in row 2,
we move {\em both} $b(s)$ and $b(t)$ into row $5$.  We move $b(y_2)$ (if it is in row 2) into row $4$, and move every other of the remaining bars of row 2 into row 3.  This 
makes {\em all} edges (except $(s,t)$) vertical, as required for a
bar-1-visibility representation.

We first handle the same-boundary crossing at $f_{st}$ (if it exists).  Assume that this is $f_{st}=\{s,t,y_2,x_\rho\}$.
Since $b(t)$ is on row 5 and $b(y_2)$ is on row 4, we can route $(t,x_\rho)$ by traversing $b(y_2)$, while $(s,y_2)$ can be routed inside 
$f_{st}$ after extending the vertices towards each other.  
As  
for planar bar visibility representations,
we also expand both $b(t)$ and $b(y_2)$ rightwards and route
outer-face edge $(t,y_2)$ at the right end so that it permits merging.
See Fig.~\ref{fig:oppositeBoundaryCrossingBar1VR}.

Now parse the opposite-boundary crossings from left to right.  
Say we handle crossing edges $(x_h,y_{z+1})$ and $(y_z,x_{h+1})$.
If $b(x_h)$ is in row $-3$, then we expand $b(x_h)$ rightwards 
beyond edge $(x_h,x_{h+1})$ and draw
$(x_h,y_{z+1})$ by traversing $b(x_{h+1})$ (which is in
a higher row since we alternate).    
See e.g.~$\{x_4,x_5,y_3,y_4\}$ in Fig.~\ref{fig:oppositeBoundaryCrossingBar1VR}.
Note that this case always applies if
$f_{st}$ contains a opposite-boundary crossing because $b(x_\rho)$ is placed in row $-3$.

If $b(x_h)$ is in row $-2$, but $b(y_z)$ is higher than $b(y_{z+1})$, then similarly we can
expand $b(y_z)$ rightwards beyond $(y_z,y_{z+1})$ and draw
$(y_z,x_{h+1})$ by traversing $b(y_{z+1})$.
See e.g.~$\{x_3,x_4,y_2,y_3\}$ in Fig.~\ref{fig:oppositeBoundaryCrossingBar1VR}.

The only case that is more difficult is when $b(x_h)$ is on row $-2$
and $y_z$ is on row $2$.   
See e.g.~$\{x_5,x_6,y_4,y_5\}$ in Fig.~\ref{fig:oppositeBoundaryCrossingBar1VR}.
Bars $b(x_h)$ and/or $b(y_z)$ may already
have been traversed when handling an opposite-boundary crossing further left,
so we must be careful not to traverse it again.  However, $b(x_h)$ is traversed  by an edge only if the last handled opposite-boundary crossing was at face $\{x_{h-1},x_h,y_{\ell},y_{\ell+1}\}$
for some $\ell<j$. In this case, none of $b(y_{\ell+1}),\dots,b(y_z)$ are traversed, because $b(y_{\ell+1})$ is not traversed at the crossing if $b(x_h)$ is, and there are no other opposite-boundary crossings involving $y_{\ell+1},\dots,y_z$.
So at least one of $b(x_h),b(y_z)$ has not been traversed. Route one of the crossing edges to traverse  that bar (after extending either $b(x_{h+1})$ or $b(y_{z+1})$ leftward). 

So in all cases we can draw one of the two crossing edges
$(x_h,y_{z+1})$ and $(y_z,x_{h+1})$ by traversing one bar
of the four vertices.  The other edge of the crossing can
be inserted within the face vertically,
by extending its endpoints towards each other.

\begin{figure}[t] 
\centering
{\includegraphics[scale=0.4,page=4]{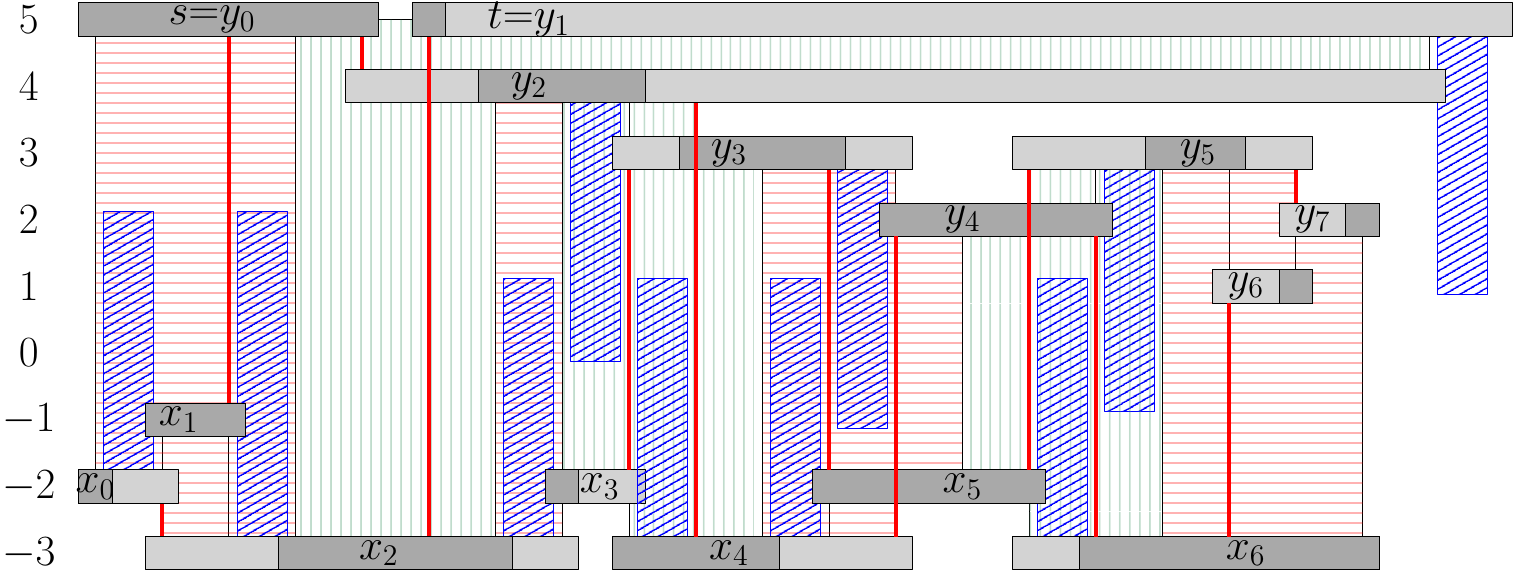}}
\caption{Handling the remaining crossings for bar-1-visibility representations.
}
\label{fig:oppositeBoundaryCrossingBar1VR}
\end{figure}

Note that by construction $s$ and $t$ are the only vertices in the topmost row, and we can expand
$b(t)$ to cover the rightmost corner.  We hence have:

\begin{lemma}
\label{lem:pathBar1VR}
For any root-to-leaf path $\pi$ in $\dual{G}$,
$P_\pi^+\setminus (s,t)$ has a bar-1-visibility representation on 8 rows that is a $C_{1,1}$-drawing.


Furthermore, the center row is empty, and
any edge that could have an attached hanging subgraph permits merging.
\end{lemma}

\subsubsection{Putting it all together.}

We now put these path-drawings together with drawings of the hanging
subgraphs by merging inward.
} 
\both{
\begin{theorem}
\label{thm:pathDrawing}
\label{thm:bar1VR}
Every outer-1-plane graph $G$ has a planar bar visibility 
representation of area $O((pw(G){+}\chi(G))n)$ and a bar-1-visibility representation
of area $O(n{\cdot}pw(G))$.
\end{theorem}
}
\onlyabstract{We can again only sketch the proof  $(\star)$.  
We first draw the planar skeleton of some outer-1-path and the hanging subgraphs much as was done for outer-planar graphs in \cite{Bie-WAOA12}.  
Based on the pathwidth (or actually the closely related parameter \emph{rooted pathwidth}), extract a root-to-leaf path $\pi$ in the dual tree $\dual{G}$ such that the rooted pathwidth of all subtrees is smaller.    Expand $P_\pi$ by adding all neighbours of $s$ to get $P_\pi^+$. Create a bar visibility representation of $\skel{P_\pi^+}$ on three rows.   
See Fig.~\ref{fig:outerpathAbstract}.   Now merge hanging subgraphs ``inward'', i.e., inside the faces of $\skel{P_\pi}$. They hence share rows and the height is only $O(1)$ more than the one of the subgraphs and works out to $O(pw(G))$.  For the merging we need $TC_{1,1}$-drawings, but with our placement of $(s,t)$ this can easily be achieved. 

However, we have not yet drawn the crossings in $P_\pi^+$.  One of each pair of crossing edges can be realized inside a face of $\skel{P_\pi^+}$.  For bar-1-visibility representations, we realize the other edges by moving vertex-bars inward or outward by one unit (plus some special handling near $s$ and $t$). After suitable lengthening of bars the other edge in a crossing can then be realized, sometimes by traversing a bar.  See Fig.~\ref{fig:draw_bar_ex}. For planar drawings, we move bars outward sufficiently far (proportionally to the number of crossings on the right) such that they can be extended rightward without intersecting other elements of the drawing.  The other edge in a crossing can then be drawn on the right.  See Fig.~\ref{fig:draw_planar_ex}.
}
\later{
\begin{proof}
We first need a small detour to explain how to choose path $\pi$.  
Pick an arbitrary reference-edge $(s,t)$, and
let $T$ be the inner tree $\dual{G}$, rooted at $f_{st}$.
We define the {\em rooted pathwidth}
$rpw(T)$ to be 1 if $T$ has exactly one leaf, and $rpw(T)=
\min_\pi \max_{T'\in \mathcal{T}(T,\pi)} \{1+rpw(T')\}$ otherwise, where
the minimum is over all root-to-leaf paths $\pi$, and $\mathcal{T}(T,\pi)$
is the set of rooted subtrees obtained by deleting the vertices of $\pi$ from
$T$.    
The path $\pi$ that achieves
the minimum is called a {\em spine} and satisfies that $rpw(T')<rpw(T)$ for
all $T'\in \mathcal{T}(T,\pi)$.  
See \cite{Bie-OPTI} for details.   

We now prove (by induction on $rpw(T)$) that we have
\begin{itemize}
\item a planar bar visibility representation 
with height at most $2rpw(T)+2\chi+1$,
\item a bar-1-visibility representation of $G\setminus (s,t)$ with  height at most $6rpw(T)+1$.
\end{itemize}
Furthermore, these are $C_{1,1}$-drawings.
In the base case, $rpw(T)=1$, so $G$ is an outer-1-path and the result holds by Lemma~\ref{lem:pathFewCrossings} and \ref{lem:pathBar1VR}  since we can delete the empty center-row.  
So assume $rpw(T)>1$, and let $\pi$ be a spine of $T$.
Apply Lemma~\ref{lem:pathFewCrossings} and \ref{lem:pathBar1VR}, respectively,
to obtain drawing $\Gamma_P$ of $P_\pi^+$.
Apply induction to any hanging subgraph $H$ to obtain its drawing $\Gamma_H$.

In $\Gamma_H$, bars $b(u),b(v)$ of the reference-edge $(u,v)$ occupy the top corners.  In $\Gamma_P$, $(u,v)$ permits merging.
If $(u,v)$ is drawn horizontally in $\Gamma_P$, then we can simply merge
the drawing inward, after adding sufficiently many new rows at the center-row.
If $(u,v)$ is drawn vertically in $\Gamma_P$, 
say with (up to symmetry) both $u,v$ above the center-row and $u$ above $v$ by $Y\geq 0$ rows,  then simply raise $u$ in $\Gamma_H$ by $Y$ rows and call the result $\Gamma_H'$.
Then the placement of $u$ and $v$ is the same in $\Gamma_P$ and $\Gamma_H'$ (up to rotation) and so again after adding new rows as needed at the center row  $\Gamma_H'$ can be merged.  

The number of new rows that we need for the subgraphs is no more than
$\max_H \{ \mathit{height}(\Gamma_H) -2\}$:   we already had the (empty)
center-row that we can use, and we do not need space row of $b(u)$ and $b(v)$ (or all other rows added for $\Gamma_H'$) since these reuse rows that already exist in $\Gamma_P$.

For bar-1-visibility representations, $\Gamma_P$ used 8 rows 
and each $\Gamma_H$ has height at most $6(rpw(T)-1)+1$ by induction
since the rooted pathwidth of $\dual{H}$ is smaller.    So the
height is at most $8+(6rpw(T)-5)-2=6rpw(T)+1$ as desired.
For bar visibility representations,
$\Gamma_P$ used $4+2\chi_P$ rows, and each hanging
subgraph has at most $\chi-\chi_P$ crossings and hence height at most
$2(rpw(T)-1)+2(\chi-\chi_P)+1$.  
So the total height is at most $2rpw(T)+2\chi+1$ as desired.

\medskip
By induction our claim holds.
It is well-known that 
$rpw(T)\leq 2pw(T)+1$ \cite{Bie-OPTI}
and that $pw(T)\leq pw(G)$ \cite{BF02}.
So we create drawings of height proportional to $pw(G)$ and $pw(G)+\chi$, respectively.
For bar-1-visibility representations, currently edge $(s,t)$
is missing, but it can be inserted horizontally, and then be made 
vertical by moving one of $s$ and $t$ up into a new row.  This proves the theorem since the width is $O(n)$
by Obs.~\ref{obs:widthHeightUpperTrivial}.
\end{proof}

} 

\iffull
\noindent By Theorems~\ref{thm:pwLowerBounds} and \ref{thm:crossingLowerBounds} the heights of these 
drawings are asymptotically optimal.  
\fi
\section{Conclusions and open problems}
\label{sec:conclusion}

In this paper, we studied visibility representations of outer-1-planar graphs.  We showed that if the embedding must be respected, then $\Omega(n^{1.5})$ area is sometimes required, and $O(n^{1.5})$ area can always be achieved.  We also studied numerous other drawing models, showing that $o(n^{1.5})$ area can be achieved as soon as we allow bends in the vertices or the edges or can change the embedding.  We also achieve optimal 
area for bar-1-visibility representations and planar visibility representations.
Following the steps of our proofs, it is clear that the drawings can be constructed in polynomial time; with more care when handing subgraph-drawings (and observing that path $\pi$ can be found in linear time \cite{biedl2021generalized,Bie-OPTI}) the run-time can be reduced to linear.
A number of open problems remain:
\begin{itemize}
\item Our drawings of height $O(n^{0.48})$ were based on the idea of so-called LR-drawings of trees \cite{Chan02}, which in turn were crucial ingredients for obtaining small embedding-preserving straight-line drawings of outer-planar graphs. With a different approach, 
Frati et al.~\cite{FPR20} achieved height $O(n^{\varepsilon})$ for drawing outer-planar graphs.   Can we achieve height $O(n^\varepsilon)$ (hence area $O(n^{1+\varepsilon})$ in some of our constructions as well?
\item Our bar-1-visibility representations do not preserve the embedding, both because the edges that go through some vertex-bar are not in the right place in the rotation, and because we merge hanging subgraphs inward.  What area can we achieve if we require the embedding to be preserved? 
\item We achieved height $O(n^{0.48})$ in complexity-1 OPVRs.  It is not hard  to achieve the optimal height $O(pw(G))$ if we allow higher complexity (complexity 4 is enough; we leave the details to the reader).  What is the status for complexity 2 or 3, can we achieve height $o(n^{0.48})$?
\end{itemize}

Finally, are there other significant subclasses of 1-planar graphs for which we can achieve $o(n^2)$-area drawings, either straight-line or visibility representations?



\bibliographystyle{splncs04.bst}
\bibliography{refs}

\clearpage
\appendix

\magicappendix

\end{document}